\documentclass[11pt]{article}
\usepackage[a4paper,margin=1in]{geometry}
\usepackage[T1]{fontenc}
\usepackage[utf8]{inputenc}
\usepackage{lmodern}
\usepackage{microtype}
\usepackage{graphicx}
\usepackage{amsmath,amssymb}
\usepackage{amsthm}
\usepackage{siunitx}
\usepackage{booktabs}
\usepackage{tabularx}
\usepackage{xcolor}
\usepackage{tikz}
\usetikzlibrary{arrows.meta}
\usepackage[hidelinks]{hyperref}

\newtheorem{theorem}{Theorem}

\newtheorem{lemmaS}{Supplementary Lemma}
\newtheorem{propositionS}{Supplementary Proposition}

\newcommand{\calX}{\mathcal{X}}
\newcommand{\calH}{\mathcal{H}}
\newcommand{\calB}{\mathcal{B}}
\newcommand{\id}{\mathbb{1}}
\newcommand{\nmax}{\bar{n}}
\newcommand{\Ncut}{N}
\DeclareMathOperator{\Tr}{Tr}

\title{Semi-device-independent quantum randomness certification in semiconductor spin-noise measurements}

\author{
Hamid Tebyanian$^{1,\ast}$, Marius Cizauskas$^{2}$, Manfred Bayer$^{2}$,\\
Marc Assmann$^{2}$, and Alex Greilich$^{2}$\\[0.6em]
\small $^{1}$\,School of Physical and Chemical Sciences, Queen Mary University of London, London, United Kingdom\\
\small $^{2}$\,Experimentelle Physik 2, Technische Universität Dortmund, 44221 Dortmund, Germany\\
\small $^{\ast}$\,Corresponding author
}

\date{\today}

\begin{document}
\maketitle

\begin{abstract}
Complex solid-state systems are promising platforms for scalable, high-bandwidth quantum random-number generation, yet certifying the quantum origin of their fluctuations remains difficult because the underlying microscopic dynamics are hard to model and validate. {Here we demonstrate semi-device-independent quantum randomness certification from semiconductor spin noise, to our knowledge the first such certificate on any spin-noise source, without relying on a microscopic model of the spin system.} {The untrusted optical source is constrained by an experimentally tested mean-photon-number bound together with a declared analogue-range and per-sample clipping ceiling, while the trusted receiver is described as a calibrated, noisy, coarse-grained homodyne measurement.} Using a semidefinite programme with rigorously controlled Fock-space truncation, we certify randomness that remains private against an adversary holding arbitrary quantum side information. Offline analysis yields certified entropy rates of $3.2$--$3.4$\,Gbit/s from a singly charged (In,Ga)As quantum-dot ensemble and $33$\,Mbit/s from $n$-GaAs in a spin-noise-matched detection mode. {This exceeds the certified entropy rate of earlier spin-noise generators by more than two orders of magnitude, and the certificate tolerates a resolved per-symbol energy contribution from the solid-state emitter itself rather than requiring a near-vacuum input.}
\end{abstract}

Random numbers underpin modern cryptography, numerical simulation and emerging quantum technologies. Quantum random-number generators exploit the intrinsic unpredictability of quantum measurements, offering an alternative to algorithms whose outputs are ultimately determined by an initial seed. However, observing a noisy quantum signal is not by itself sufficient to guarantee secure randomness. The amount of randomness that can be extracted depends on which parts of the apparatus are trusted, how accurately they are characterised and what information might be available to an adversary. Establishing these assumptions explicitly, while retaining the high generation rates needed for practical applications, remains a central challenge for quantum random-number generation~\cite{Herrero-Collantes2017,Ma2016}.

Different approaches make different compromises between security and experimental complexity. Fully device-independent protocols certify randomness without relying on detailed models of the source or detector, but require loophole-free Bell tests and therefore remain difficult to operate at the bandwidths of conventional optical receivers~\cite{Pironio2010, Liu2018, Liu2021, Shalm2021}. At the other extreme, device-dependent generators can reach very high rates, but their security relies on an accurate physical description of the source and measurement apparatus~\cite{Gabriel2010,Haw2015,Gehring2021}. Semi-device-independent protocols provide an intermediate route: selected properties of the apparatus are calibrated and trusted, whereas the remaining components are treated adversarially under a small number of experimentally testable constraints.
For continuous-variable optical measurements, a particularly natural constraint is the mean optical energy of the incoming signal~\cite{VanHimbeeck2017,Rusca2020,Avesani2020,Tebyanian2021}.

Spin fluctuations in semiconductors offer an attractive, but so far largely unexplored, platform for such certified randomness generation. Even at thermal equilibrium, the finite number of electron spins within an observation volume produces stochastic fluctuations of the collective spin polarisation~\cite{Bloch1946}. These fluctuations can be detected optically through the small Faraday rotation that they impart to a probe beam~\cite{Aleksandrov1981,Crooker2004, Oestreich2005}. Spin-noise spectroscopy has become an established method for studying semiconductor spin dynamics without deliberately polarizing the spin system~\cite{Zapasskii2013,Muller2010SpinNoiseReview}, and spin noise has previously been used to generate random numbers in atomic vapours~\cite{Katsoprinakis2008,Dellis2024}. Those demonstrations, however, relied on device-dependent descriptions of the source and operated at rates limited by the comparatively narrow spin-noise bandwidth. Whether solid-state spin fluctuations can support high-rate randomness that remains certifiable without a microscopic model of the spin system has remained an open question.

{Here we demonstrate semi-device-independent quantum randomness certification using equilibrium spin noise in a semiconductor system. Earlier spin-noise generators, in atomic vapours, were device-dependent, and earlier semi-device-independent generators used vacuum or coherent light; to our knowledge no semi-device-independent certificate has previously been obtained on a spin-noise source of any kind.} A balanced homodyne receiver measures the weak optical field generated in the polarisation orthogonal to the incident probe. We treat the local oscillator, interferometer, detector, digitiser, and binning procedure as a trusted and calibrated measurement apparatus. By contrast, no microscopic model is assigned to the semiconductor source: it may produce any optical state, potentially correlated with arbitrary quantum side information held by an adversary, provided that it satisfies a tested mean-photon-number bound and a declared analogue-range condition. The observed outcomes are described by the calibrated, noisy, and coarse-grained homodyne measurement. We then use a semidefinite programme (SDP), together with a rigorously controlled truncation of the optical Fock space, to bound the conditional min-entropy of the measurement outcomes.

An important feature of this construction is that the measured spin-noise signal cannot artificially increase the certified randomness. Additional source energy enlarges the set of states available to an adversary and therefore lowers, rather than raises, the entropy bound. The spin-noise measurement consequently tests whether substantial randomness can still be certified while the fluctuating solid-state signal is present and experimentally resolved. The source-energy bound can either be supplied as an independent trusted parameter or determined within the protocol by measurements of two optical quadratures. Measuring both quadratures is essential because a source could appear vacuum-like in one quadrature while carrying unrestricted energy in the orthogonal one.

We apply this framework to two complementary spin-noise systems. In an $n$-doped GaAs epilayer, the spin fluctuations form a narrow Larmor resonance~\cite{Petrov2018}, yielding a certified generation rate of 33\,Mbit/s for the 131\,ms acquisition. In a singly charged (In,Ga)As quantum-dot ensemble, hyperfine and electron-$g$-factor variations broaden the spin fluctuations across the receiver bandwidth~\cite{Crooker10,Kamenskii2020}. {This inhomogeneous broadening distributes the spin noise over many effective temporal modes and enables certified rates of $3.2 - 3.4$\,Gbit/s, more than two orders of magnitude above the spin-noise generators demonstrated so far.} Together with a broadband GaAs measurement serving as a vacuum-limited control, these experiments establish semiconductor spin noise as a high-rate entropy source compatible with source-adversarial randomness certification. Although the present demonstration uses offline analysis of fixed-phase records, the same framework provides a route towards real-time composable operation using randomized phase switching between separate test and generation rounds. {We present and analyse two protocols: in protocol~1 the energy bound is supplied as a trusted external input, whereas in protocol~2 it is measured within the protocol by switching the receiver between two orthogonal homodyne phases, at the cost of one additional statistical confidence parameter. Throughout, certified denotes certification relative to the calibrated receiver model and the explicit protocol assumptions summarised in Methods. Supplementary Table~S3 positions this work against earlier spin-noise and continuous-variable generators; the security theorems and proofs, together with the calibration records, are given in Methods and Supplementary Notes~1--8.}

\section*{Results}

\subsection*{Measurement principle and detection regimes}

In a finite spin ensemble, the populations of oppositely oriented spins never balance exactly but fluctuate spontaneously in time. When a polarized probe beam is transmitted through the spin ensemble, the fluctuating magnetization produces a weak optical field in the polarisation orthogonal to the probe polarization through the magneto-optical interaction~\cite{Gorbovitskii1983, Scalbert2019}. Balanced homodyne detection converts this field into a rapidly varying electrical signal. In conventional spin-noise spectroscopy, the power spectrum of this signal is used to extract spin-dynamical parameters~\cite{Zapasskii2013,Muller2010SpinNoiseReview}. Here, the spectral distribution instead determines how the signal occupies effective temporal modes and hence the bandwidth available for randomness generation.

Figure~\ref{fig:concept} summarizes the three detection regimes studied experimentally: broadband GaAs detection, a mode matched to a narrow GaAs Larmor resonance, and broadband detection of a quantum-dot ensemble. All three scenarios use the same calibrated receiver model and certification procedure, so differences in the certified rates can be traced to temporal-mode selection and usable bandwidth rather than to an assumed correspondence between spin-noise variance and intrinsic randomness.

\subsection*{Certification architecture}

The certification model divides the generator into an untrusted source and a trusted receiver, as illustrated in Fig.~\ref{fig:concept}B. The untrusted part is the weak signal mode emerging from the sample in the orthogonal-polarisation channel. We denote this mode by $A$ and allow it to be correlated with an arbitrary quantum system $E$ held by an adversary; together, they are described by the joint state $\rho_{AE}$. No microscopic model of the semiconductor source is imposed.

The source obeys $\Tr(\rho_A\hat n)\le\nmax$, where $\hat n=a^\dagger a$ is the photon-number operator of one effective detected mode and $\nmax$ is its mean-energy upper bound in photons per mode.

The trusted device comprises the local oscillator (LO), interferometer, balanced detector, analogue and digital electronics, phase-control, and the classical post-processing that maps the calibrated output to discrete bins. The receiver is modeled as a noisy binned homodyne positive operator-valued measure (POVM) $\{\Pi_x\}$~\cite{Davies1970, Busch1996}. Three calibrated parameters describe its response: the trusted-noise variance $\sigma^2$, the bin width $\delta$, and the boundary $x_{\max}$ of the interior detection range. The numerical values for the three measurement regimes are collected in Supplementary Table~S2. We handle finite receiver range and saturation as separate conditions, preventing out-of-range signals from being accepted as ordinary generation outcomes.

Given the calibrated POVM and the energy bound, the SDP searches for the state that would make the digitized outcome easiest for an adversary to predict~\cite{Koenig2009,Tavakoli2024}. This worst-case optimization yields the certified entropy per retained symbol, $h^*_{\mathrm{cert}}(\nmax)$, defined in Eq.~\eqref{eq:hcert}. The calculation is performed in a truncated Fock space, with an explicit correction for all neglected higher-photon-number states. It therefore remains valid against arbitrary quantum side information and does not identify the observed spin-noise variance itself with private entropy.

We use two protocols to establish the energy bound. In Protocol~1, $\nmax$ is supplied as an independently validated input. In Protocol~2, it is estimated within the protocol from measurements of two orthogonal homodyne quadratures, using $X_0^2+X_{\pi/2}^2=2\hat n+\id$.

Both phases are needed because a measurement of one quadrature alone cannot exclude energy concentrated in the conjugate quadrature. Finite-sample statistics and corrections for saturation, phase-calibration uncertainty, and local-oscillator instability are included in the resulting bound.

\begin{figure}[t]
\centering
\includegraphics[width=\linewidth]{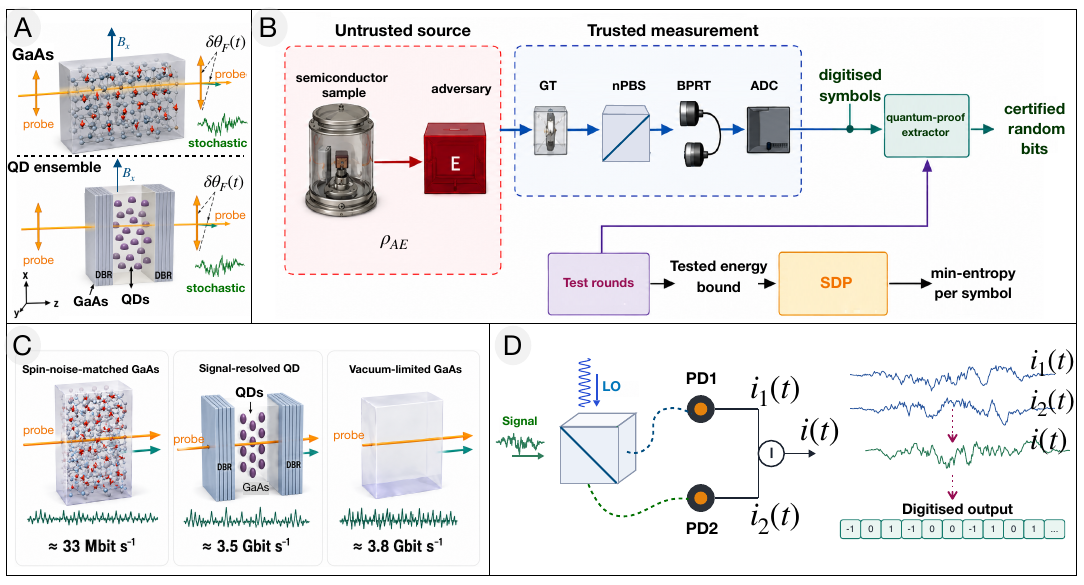}
\caption{{\textbf{Source-adversarial certification of semiconductor spin noise.} \textbf{A}, Sources: bulk $n$-GaAs and a charged $(\mathrm{In},\mathrm{Ga})\mathrm{As}$ quantum-dot ensemble in a microcavity, in transverse field $B_x$; probe Faraday rotation $\delta\theta_F(t)$ from equilibrium spin fluctuations generates the orthogonally polarised signal. \textbf{B}, Trust boundary. The source and any adversary $\rho_{AE}$ with memory $E$ are untrusted, bounded only by the tested mean-photon-number and declared analogue range; the receiver (GT analyzer, nPBS, balanced photoreceiver BPR, ADC) is trusted and calibrated. Two-phase test rounds fix $\nmax^{\rm test}$; the truncated-Fock SDP returns $h^*_{\rm cert}$, which sets the extractor output. \textbf{C}, Three regimes (Table~\ref{tab:results}): spin-noise-matched GaAs, signal-resolved broadband QD, and the vacuum-limited GaAs benchmark. \textbf{D}, Balanced homodyne: signal and LO interfere; the photocurrent difference $i_-(t)$ is digitized to the certified symbol stream.}}
\label{fig:concept}
\end{figure}

Extending the single-mode certificate to the complete record assumes that the retained effective modes are independent and identically distributed (IID). Under this assumption, Theorems~\ref{thm:protocol1} and \ref{thm:protocol2} (see Methods) determine the extractable entropy from $n$ retained rounds. Finally, a quantum-proof two-universal hash compresses the raw record into a shorter, nearly uniform bit string that remains private from the side-information system $E$~\cite{Koenig2009}. Figure~\ref{fig:concept} summarizes this trust structure and the three measurement regimes that instantiate it.

\subsection*{Spin-noise platforms}
We perform balanced homodyne detection of the weak field generated in the polarisation channel orthogonal to a linearly polarised probe transmitted through the sample (Fig.~\ref{fig:setup}). Equilibrium electron spin fluctuations produce stochastic Faraday rotation, thereby generating the orthogonally polarised signal field. The transmitted probe is rejected by a high-extinction Glan-Thompson analyzer, and the signal interferes with an LO of the same polarisation on a 50/50 beamsplitter. The LO never interacts with the spin system, and the balanced difference signal selects one optical quadrature while suppressing common-mode LO intensity noise. A piezo-actuated mirror stabilizes the LO--signal phase with a proportional--integral--derivative lock; the amplification chain and digitiser are specified in Methods.

\begin{figure}[t]
    \centering
    \includegraphics[width=\linewidth]{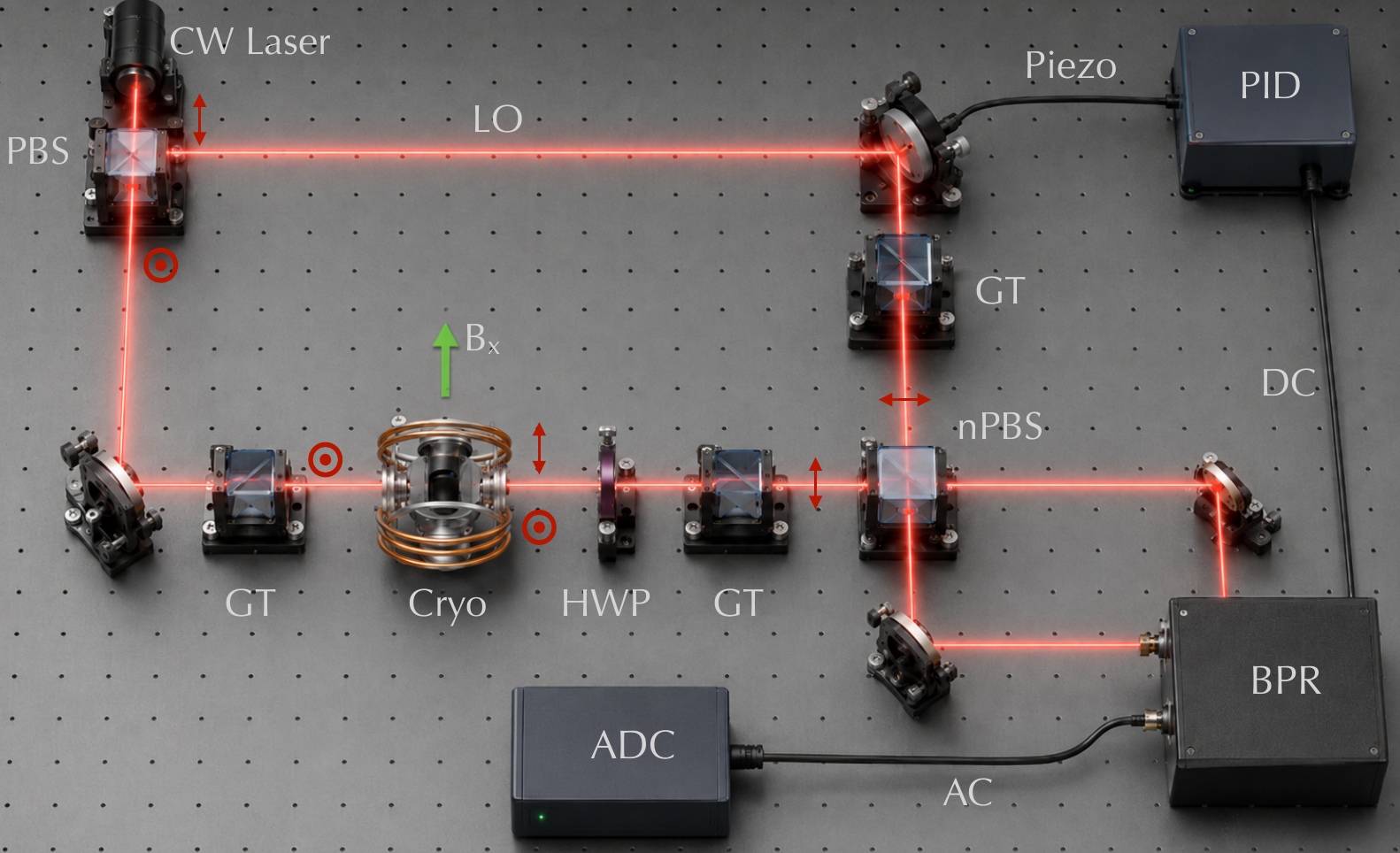}
    \caption{\textbf{Experimental setup.} A continuous-wave laser is split into a probe, which measures spin noise via Faraday rotation in the sample (held at \SI{6}{K} in a transverse magnetic field $B_x$), and a local oscillator (LO) routed around the sample. The orthogonal-polarisation component after the sample interferes with the LO on a 50/50 non-polarizing beamsplitter (nPBS) and is detected by a balanced photoreceiver (BPR). The differential DC output drives the phase lock; the AC component is amplified, filtered, and digitized.}
    \label{fig:setup}
\end{figure}

For $n$-doped GaAs (probe \SI{846.1}{nm}, $B_x=\SI{20}{mT}$, \SI{6}{K}, probe powers \SIrange{10}{60}{mW}, LO \SI{2}{mW} per detector arm~\cite{Petrov2018}) the spin noise is concentrated in a narrow Larmor resonance. Block-averaging the \SI{1.31}{ms} digitiser records over $1310$ segments of 1\,$\mu$s resolves the few-percent excess above the shot-noise floor as a Lorentzian centred at $f_0=\SI{126.0}{MHz}$, with half-width \SI{5.15}{MHz} (Supplementary Fig.~S1). This frequency band defines the spin-noise-matched mode. {The stored one-second on-board FFT spectra, acquired through the separate 8-bit FPGA monitoring path, show the same phase-selective narrow resonance (Fig.~\ref{fig:regimes}A). These spectra illustrate the spectral response only; the spin-noise-matched band used for certification is defined from the digitiser records themselves.}

A second sample on the same receiver, a singly charged $(\mathrm{In},\mathrm{Ga})\mathrm{As}/\mathrm{GaAs}$ quantum-dot ensemble in a planar microcavity ($B_x=\SI{40}{mT}$, {\SI{6.3}{K}}~\cite{Kamenskii2020}), exhibits spin noise whose inhomogeneous broadening, arising from hyperfine and $g$-factor dispersion across the ensemble, fills the entire analogue detection window (Fig.~\ref{fig:regimes}B and Supplementary Fig.~S2). Its per-sample second moment already carries the spin-noise signal. Figure~\ref{fig:landscape} maps the measured spin-noise excess of both platforms against frequency and probe power.

{The regime separation is already visible in the raw broadband digitiser records. Figure~\ref{fig:regimes}C plots the full-band contrast between the two saved lock settings,
\begin{equation}
\label{eq:contrast}
\Delta V(P)=100\,\frac{\operatorname{Var}(Y_{\mathrm{max\text{-}SN}})-\operatorname{Var}(Y_{\mathrm{min\text{-}SN}})}{\operatorname{Var}(Y_{\rm LO})-\operatorname{Var}(Y_{\rm dark})},
\end{equation}
computed from the stored per-record variances. The quantum-dot contrast grows to $14.6\%$ of the shot-noise reference at \SI{30}{mW}, whereas the GaAs full-band contrast remains below $0.7\%$ in magnitude at every power because the narrow line is diluted across the receiver band. $\Delta V$ is a descriptive comparison rather than the security estimator. The tested bound of Eq.~\eqref{eq:nmax-tested-general} uses both absolute second moments together with their statistical, saturation, phase, and LO corrections. The maximum- and minimum-spin-noise labels are acquisition labels taken from the lock metadata, not an independent phase calibration.}

\begin{figure*}[t]
    \centering
    \includegraphics[width=\linewidth]{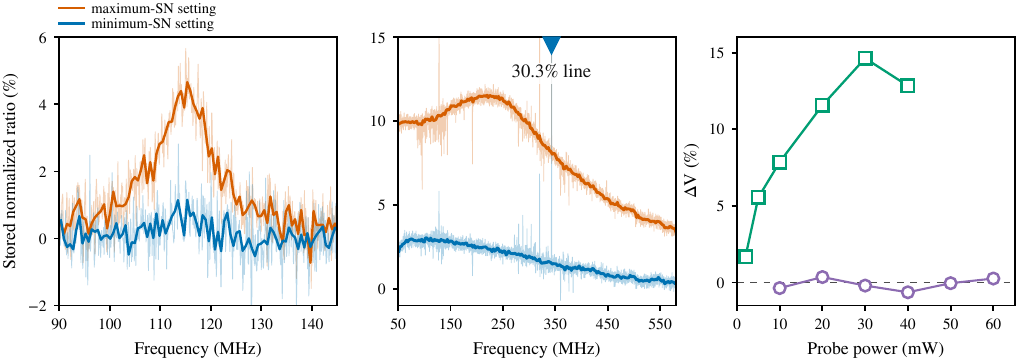}
    \caption{{\textbf{Two spectral regimes in the stored spin-noise data.} \textbf{A}, Normalised FPGA spectra stored for $n$-GaAs at \SI{40}{mW} for the two saved lock settings (maximum- and minimum-spin-noise), showing the phase-selective narrow spin-noise resonance. \textbf{B}, Corresponding quantum-dot spectra at \SI{20}{mW}: the spin-noise excess fills the analysed receiver band; the marker indicates an isolated $30.3\%$ technical line outside the displayed range. Thin curves show the stored ratios; thick curves show frequency-binned medians for visibility. \textbf{C}, Full-band variance contrast $\Delta V$ of Eq.~\eqref{eq:contrast} from the \SI{1.31}{ms} digitiser records at each recorded probe power; error bars are temporal-block standard errors from 128 contiguous blocks within each record and describe within-record variability only.}}
    \label{fig:regimes}
\end{figure*}

\begin{figure*}[t]
\centering
\includegraphics[width=\linewidth]{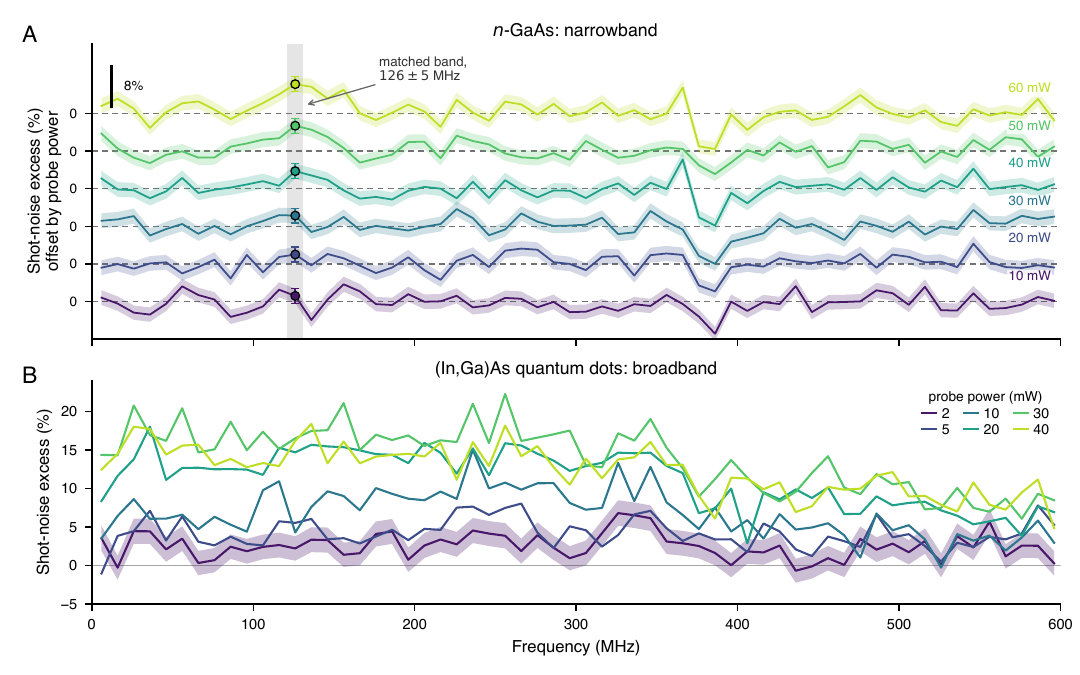}
\caption{\textbf{Narrowband and broadband spin-noise excess.} Spectral excess $100\,(S_{\mathrm{sig}}-S_{\mathrm{LO}})/S_{\mathrm{LO}}$, at the maximum-spin-noise lock setting, obtained from Hann-windowed \SI{1}{\micro\second} blocks of the \SI{1.31}{ms} records and collected into \SI{10}{MHz} bins. \textbf{A}, $n$-GaAs. Traces are offset by $7\%$ per probe power; dashed lines indicate their zero levels. The narrow Larmor signal increases with probe power within the predeclared matched band, \SIrange{121}{131}{MHz}, shaded in grey. Filled circles give the matched-band excess. \textbf{B}, $(\mathrm{In},\mathrm{Ga})\mathrm{As}$ quantum dots, shown without offsets. The spin-noise excess extends across the receiver band and reaches its largest broadband value near \SI{30}{mW}. Shaded bands and error bars denote $\pm1$ standard error, including the Hann-window correlation correction. These spectra are descriptive and are not the energy estimator used for certification.}
\label{fig:landscape}
\end{figure*}

\subsection*{Certified entropy and rates}
The receiver was calibrated from dark and LO-only traces (Supplementary Table~S2). The quantity $h^*_{\mathrm{cert}}(\nmax)$ was evaluated from the dual programme at Fock cutoff $\Ncut=200$, including the truncation correction of Eq.~\eqref{eq:hcert}, and the numerical solution was converted into a {feasible dual point with an a-posteriori residual check} (Supplementary Note~4). Rates are quoted as $R_{\rm cert}=h^*_{\rm cert}f_s$ at extractor error $\varepsilon_{\rm ext}=10^{-10}$; the extractor finite-size penalty is below the quoted rounding in every row (Methods). The operating bin width is $K=64$ digitiser codes; a sweep of $K$ confirms that this is a fixed operating choice rather than a tuned post-selection parameter (Supplementary Fig.~S6). Table~\ref{tab:results} lists the certified results for the matched GaAs, broadband quantum-dot and broadband GaAs measurements. {Each rate is the certified per-mode min-entropy multiplied by the declared effective-mode rate; no extracted composable bitstream is claimed for these offline records.}

\begin{table}[t]
\caption{Certified per-symbol min-entropy and generation rate at $K=64$, $\Ncut=200$, $\varepsilon_{\rm test}=10^{-6}$. For the GaAs matched-mode rows, values obtained from the tested energy bound are followed in parentheses by the corresponding point-estimate values. Matched-mode rates use the conservative effective-mode symbol rate $f_s=\SI{10}{MHz}$; broadband rates use the receiver-bandwidth-limited $f_s=\SI{1.0}{GHz}$. $R_{\rm cert}$ is the certified  generation rate before diversion of test rounds; a real-time energy test diverting a fraction $p_{\rm test}$ of rounds scales it by $(1-p_{\rm test})$. }
\label{tab:results}
\centering
\begin{tabular}{lccc}
\toprule
 & $\nmax^{\rm test}$ & $h^*_{\rm cert}$ (bit) & $R_{\rm cert}$\\
\midrule
GaAs matched, $10$ mW & $0.077\,(0.013)$ & $2.87\,(3.24)$ & $29\,(32)$ Mbit/s\\
GaAs matched, $40$ mW & $0.083\,(0.020)$ & $2.84\,(3.18)$ & $28\,(32)$ Mbit/s\\
GaAs matched, $60$ mW & $0.093\,(0.029)$ & $2.81\,(3.10)$ & $28\,(31)$ Mbit/s\\
GaAs matched, $40$ mW (\SI{131}{ms}) & $0.011\,(0.008)$ & $3.27\,(3.32)$ & $33\,(33)$ Mbit/s\\
GaAs broadband, $10$--$60$ mW & $0.030$--$0.033$ & $3.52$--$3.55$ & $3.52$--$3.55$ Gbit/s\\
QD broadband, $2$ mW    & $0.052$ & $3.36$ & $3.36$ Gbit/s\\
QD broadband, $20$ mW   & $0.084$ & $3.20$ & $3.20$ Gbit/s\\
QD broadband, $30$ mW   & $0.095$ & $3.15$ & $3.15$ Gbit/s\\
QD broadband, $40$ mW   & $0.089$ & $3.18$ & $3.18$ Gbit/s\\
\bottomrule
\end{tabular}
\end{table}

In the GaAs matched mode, the record is bandpassed to $126\pm\SI{5}{MHz}$, so the energy test probes the spin-noise band directly. The point-estimate photon number per symbol grows from $0.013$ to $0.029$ with probe power, tracking the in-band excess of Fig.~\ref{fig:landscape}A. Under Protocol~2 the entropy bound uses the tested upper bound $\nmax^{\rm test}$ rather than the point estimate. On the \SI{1.31}{ms} records this bound is statistics-limited, the finite-statistics and range corrections contributing about $0.063$ photon, so that $\nmax^{\rm test}=0.077$--$0.093$ and the certified rate is near \SI{28}{Mbit/s}. A \SI{131}{ms} record at \SI{40}{mW} shows the contraction of this margin with record length (Fig.~\ref{fig:scaling}): the tested bound reaches $\nmax^{\rm test}=0.011$, compared with the point estimate $0.008$, and the certified rate rises to \SI{33}{Mbit/s}. The support-disjoint evaluation of the long record gives $\nmax^{\rm test}=0.020$ and $h^*_{\rm cert}=3.17$ bits per symbol, one step below at \SI{32}{Mbit/s}. The decomposition of the tested bound into resolved source energy and statistical margin is shown in Supplementary Fig.~S5.

{In the quantum-dot broadband mode, the tested bound contains both the certified randomness and a resolved source-energy contribution for each effective symbol.} In this mode the spin-noise spectrum fills the detection window, so every digitised sample carries source energy and no matched filter is needed: the effective mode rate is set by the trusted receiver bandwidth, $f_s=\SI{1.0}{GHz}$. The tested photon number grows from $0.052$ at \SI{2}{mW} to $0.095$ at \SI{30}{mW}. The corresponding point estimate, $0.015$--$0.059$, constitutes a resolved fraction of the tested bound rather than being set by the statistical floor. The certified rate is $3.2$--$3.4$~Gbit\,s$^{-1}$. {The GaAs broadband measurement on the same receiver shows a full-band contrast below $0.7\%$, whereas the quantum-dot contrast reaches $14.6\%$. The excess is strongly suppressed at the minimum-spin-noise phase and varies non-monotonically with probe power, peaking near \SI{30}{mW}; both behaviours differ from phase-independent electrical pickup and simple monotonic probe leakage. Spin-noise spectroscopy of the same trion-resonant sample was established independently in Ref.~\cite{Kamenskii2020}. This attribution is not required for security:} narrow technical lines from residual laser and RF pickup enter only through the measured second moment, so they can only raise $\nmax^{\rm test}$ and lower the certificate, making these rates conservative. The GaAs broadband mode completes the picture as a control: across the full bandwidth the GaAs spin-noise contribution sits below the per-sample detection threshold, the tested energy is essentially the statistical floor ($0.030$--$0.033$), and the resulting $3.52$--$3.55$~Gbit\,s$^{-1}$ is a vacuum-limited benchmark of the calibrated receiver rather than a spin-noise-attributed rate. {The calibrated quadrature statistics determine the entropy, whereas any source energy, including spin noise, enlarges the adversary's feasible set and lowers the certificate.}

Figure~\ref{fig:master} collects all operating points on the certified entropy--energy curves of the three calibrated receiver modes: one SDP construction governs every regime, with the tested energy fixing the certificate. Tightening $\varepsilon_{\rm test}$ from $10^{-6}$ to $10^{-12}$ lowers $h^*_{\rm cert}$ by less than $0.1$~bit on the broadband and long-record rows; on the three short matched rows the dominant range term of the empirical-Bernstein correction is linear in $\ln(1/\varepsilon_{\rm test})$, giving a cost of about $0.13$~bit, so the results are stable under stricter security parameters.

\begin{figure}[t]
    \centering
    \includegraphics[width=0.85\linewidth]{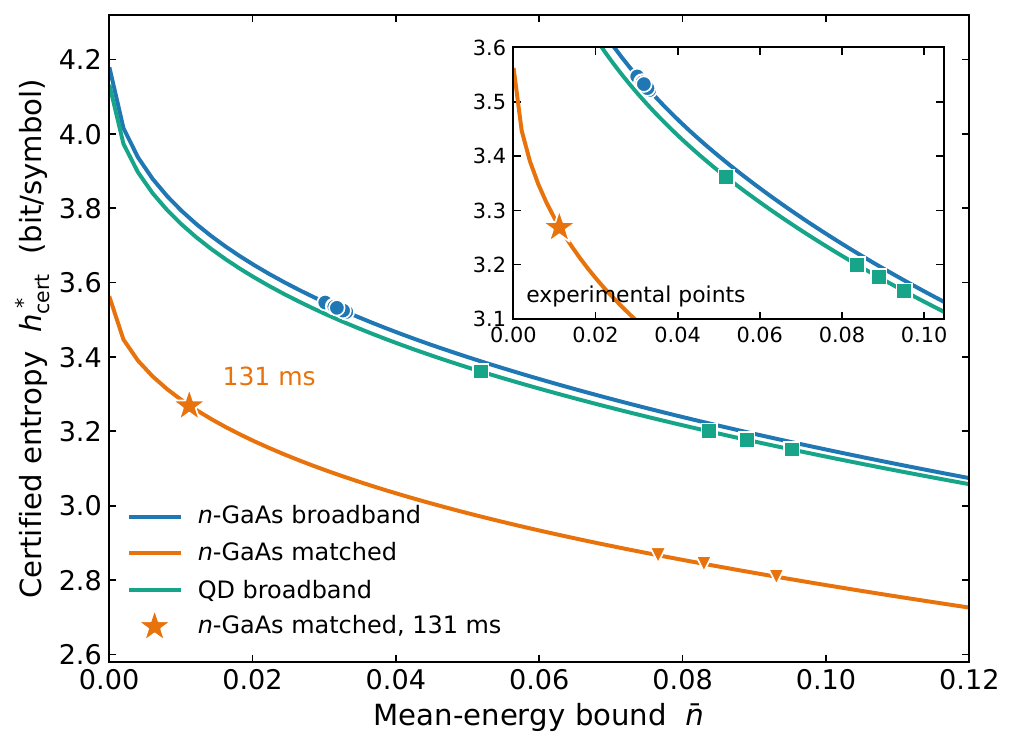}
    \caption{\textbf{Certified entropy versus tested source energy.} Certified per-symbol min-entropy $h^*_{\rm cert}$ as a function of the mean-energy bound $\nmax$ for the three calibrated receiver modes, all obtained using the same truncated-Fock SDP with the corresponding noisy-homodyne POVM. Markers: experimental tested-energy operating points; star: the \SI{131}{ms} GaAs matched acquisition; inset: enlarged experimental region. GaAs broadband points form the vacuum-limited receiver benchmark, GaAs matched points the narrowband spin-noise-attributed certificates, and quantum-dot points the broadband signal-resolved certificates.}
    \label{fig:master}
\end{figure}

\begin{figure}[t]
    \centering
    \includegraphics[width=0.85\linewidth]{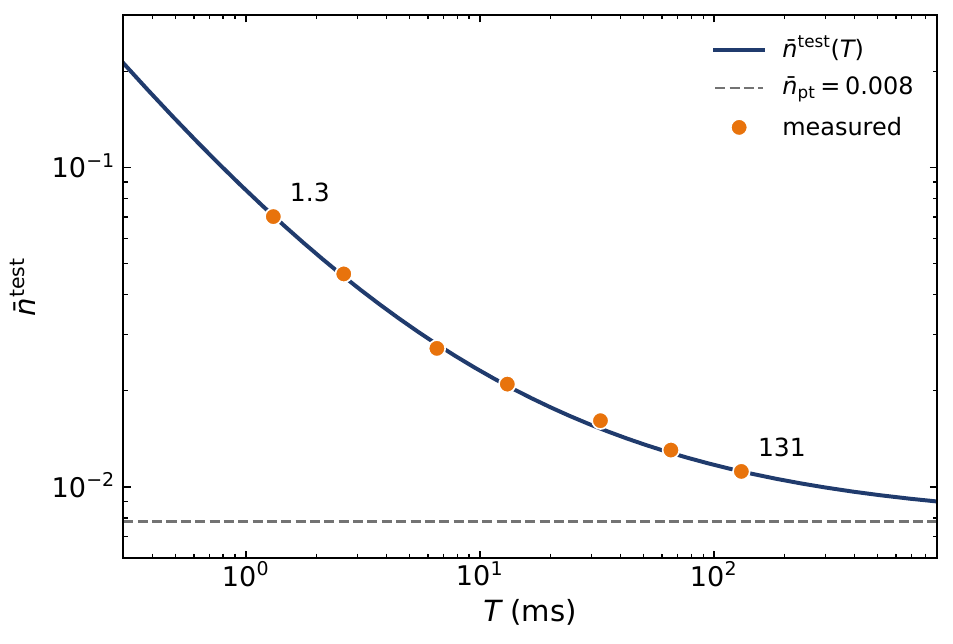}
    \caption{\textbf{Finite-statistics scaling of the energy test.} Matched-mode tested photon-number bound $\nmax^{\rm test}$ at \SI{40}{mW} versus record length $T$. {The empirical-Bernstein penalty contracts with record length — the $T^{-1}$ range term dominating at \SI{1.31}{ms}, the $T^{-1/2}$ variance term dominating at long $T$ — until the signal-resolved floor $\nmax_{\rm pt}\approx0.008$ of the long acquisition is reached.} Markers: the energy test evaluated on prefixes of the \SI{131}{ms} acquisition, from \SI{1.31}{ms} to the full record; curve: the empirical-Bernstein bound with $n_\phi\propto T$ at the moments of the full record.}
    \label{fig:scaling}
\end{figure}

\section*{Discussion}
Equilibrium semiconductor spin noise can therefore be incorporated into a source-adversarial SDI certification framework. Importantly, the certification does not assume that the spin noise itself is trusted or intrinsically random: its measured energy is treated as source-controlled and therefore only enlarges the adversary's feasible set. The trusted receiver is described by a calibrated, coarse-grained noisy-homodyne POVM; the untrusted source is constrained only by a tested mean-energy bound and a declared analogue range; and a finite truncated-Fock SDP with an explicit truncation certificate yields the min-entropy. The two platforms realise complementary regimes — a narrowband certificate in which the GaAs spin-noise energy is exposed by a matched filter, and a broadband certificate in which the quantum-dot ensemble's source energy is resolved on every sample across the receiver bandwidth — with the full-band GaAs measurement serving as the near-vacuum benchmark of the same receiver. {Of the three regimes, the quantum-dot result provides the central experimental demonstration: to our knowledge, it is the first randomness certificate in which a resolved per-symbol energy contribution from a condensed-matter source is explicitly tolerated within the security bound. The matched GaAs and broadband GaAs measurements bracket this regime from the narrowband signal-resolved and near-vacuum limits, respectively.}

The higher effective-mode rate of the solid-state platform comes from the spectral width of the spin fluctuations; the trust model is unchanged. In the cited vapour demonstrations~\cite{Katsoprinakis2008,Dellis2024} the kilohertz-scale spin-relaxation linewidth sets the source correlation time and thereby limits the number of effectively independent temporal modes per second at the source, and the demonstrated rates sit at the values summarised in Supplementary Table~S3 regardless of receiver bandwidth. In the semiconductor platforms measured here, the source bandwidth is instead comparable to, or larger than, the relevant detection bandwidth. The $n$-GaAs resonance is already more than three orders of magnitude broader than in the cited vapour systems, while in the quantum-dot ensemble hyperfine- and $g$-factor-dispersion broaden the spin noise across the full analogue detection window, so the measured quadrature correlations decay on the receiver-bandwidth timescale, opening the gigabit effective-mode regime for a spin-noise source. Inhomogeneous broadening limits the use of quantum-dot ensembles as coherent spin memories, but here it spreads the spin noise across the receiver band and increases the available effective-mode rate. The experiment still requires cryogenic operation, although temperature is not part of the security model. {The method builds on a standard spin-noise-spectroscopy receiver, supplemented by the trusted-detector calibration, energy testing and SDP analysis required for certification. The same certification framework can be applied to other homodyne records when a valid two-quadrature energy bound is available.}

{The present certificates rely on several experimental assumptions: IID effective modes at the stated decorrelation spacing, stationarity between sequential phase records, exact orthogonality of the locked phase pair ($\bar c=0$), the declared analogue-range and saturation bounds, and private receiver noise. Several of these assumptions can be replaced by directly testable conditions in a future implementation. An independently calibrated phase reference replaces the $\bar c=0$ declaration; a $2^\circ$ error would reduce the certified entropy by $0.07$--$0.17$~bit (Supplementary Note~6). An orthonormal temporal-mode construction replaces the decorrelation-spacing convention, especially in the matched mode where adjacent filtered outputs share raw samples (Supplementary Note~7), and a hardware overflow monitor makes the saturation bound testable during operation.} Real-time composable operation requires, in addition, within-run randomised phase switching with disjoint test and generation rounds: the piezo lock must be replaced by an electro-optic or integrated phase modulator switched per round, with guard intervals excluded and the switch record kept in the public transcript. {The real-time acquisition and Toeplitz-extraction chain demonstrated by this collaboration at $33.9$~Gbit\,s$^{-1}$~\cite{Cizauskas2026} already demonstrates the required real-time processing capability. That generator certifies a vacuum state through a heterodyne resolution bound, an established protocol class operated at higher speed; the present work instead moves the source class, certifying a solid-state emitter whose own energy is resolved inside the bound, and the two advances are therefore not comparable on rate.} The certificate itself is unchanged by interleaving, since Protocol~2 already assigns rounds independently of the source. What remains open is the non-IID transfer of the energy test to the generation rounds. The correlated-round techniques of Ref.~\cite{Lu2026}, developed for a ternary-input heterodyne protocol, would need to be adapted to the present discretised-homodyne POVM. A composable treatment would additionally require an appropriate min-tradeoff function for the energy-constrained measurement~\cite{DFR2020,Metger2022,Wolf2006}. {A deployed, composably secure gigabit-class solid-state SDI generator therefore still requires randomised phase switching, disjoint test and generation rounds, and a non-IID security proof.}

\section*{Methods}

\subsection*{QRNG model and trusted receiver}
The source emits a sequence of bosonic modes with annihilation operators $a_i=\int u_i^*(t)\hat b_s(t)\,dt$ defined by temporal mode functions $u_i(t)$ on the $i$-th sampling window; for non-overlapping normalised windows $[a_i,a_j^\dagger]=\delta_{ij}$, and each round carries the Fock space $\calH$ with number operator $\hat n=a^\dagger a$. The adversary holds an arbitrary $\rho_{AE}$ on $\calH\otimes\calH_E$ with $E$ unrestricted, and the only source constraint entering the optimisation is
\begin{equation}
\label{eq:energy-constraint}
\Tr(\rho_A\hat n)\le\nmax,\qquad \rho_A=\Tr_E\rho_{AE}.
\end{equation}
{The trusted receiver defines the complete detected spatiotemporal-polarisation mode: every optical or electrical degree of freedom able to affect $V_\phi$ is rejected by a trusted mode filter, contained in the mode $a$, or contained in the trusted noise $W$, so no auxiliary source mode outside this model couples to the digitised output and evades the constraint on $\Tr(\rho_A\hat n)$.}
The receiver produces $V_\phi=X_\phi+W$ with $X_\phi=(ae^{-i\phi}+a^\dagger e^{i\phi})/\sqrt2$, where $W$ is trusted receiver noise of zero mean and calibrated variance $\sigma^2$, independent of the source state {and, in the security model, private: uncorrelated with $E$ and independent across rounds}. The digitiser returns $Y_\phi=\operatorname{clip}(V_\phi,[-Y_{\min},Y_{\max}])$ and bins it into $x\in\calX$ with uniform width $\delta$ and interior cutoff $x_{\max}=31$. {The SDP is defined not for this clipped map but for the ideal unclipped noisy-homodyne binning, whose POVM is}
\begin{equation}
\label{eq:povm-def}
\Pi_x=\int_{\calB_x}\!\!\int_{\mathbb R}G_\sigma(y-y')\,|y'\rangle\langle y'|\,dy'\,dy,
\end{equation}
with $G_\sigma$ the Gaussian noise kernel and $|y'\rangle$ the generalised $X_0$ eigenstates, fixed completely by $(\sigma,\delta,x_{\max})$. A declared analogue-range model bounds the second moment hidden by saturation, $\mathbb E[V_\phi^2-Y_\phi^2]\le\beta_\phi^{\rm sat}\le(L^2-Y_{\min}^2)\delta_{\rm clip}^{\max}$, with $L$ the declared range scale and $\delta_{\rm clip}^{\max}$ the declared ceiling on the per-sample clipping probability; this and the range model are trusted operating assumptions of the receiver, not quantities inferred from the observed clipping frequency, and near-rail clips are charged to the adversary through the corrections $q_{\rm bb}$ (broadband) and $q_{\rm clip}^{\rm match}$ (matched mode) detailed in Supplementary Notes~1 and 5. {Two logically separate receiver premises enter: the Gaussian law for $W$ is the calibrated in-range noise model of the POVM, while the range premise applies to accepted events only. Samples with $|V_\phi|>L$ are overflow events, counted in the abort probability rather than described by the Gaussian model, so on every non-abort sample $|V_\phi|\le L$ holds by the acceptance rule and supplies the $L^2$ factor; the Gaussian tail beyond $L$ contributes only to the abort probability. Failure of the declared range premise is assigned probability at most $\varepsilon_{\rm range}$ and enters the abort-inclusive composable error budget alongside $\varepsilon_{\rm test}$. The range premise is therefore enforced by the abort, not asserted of the unbounded noise law. A clipped ADC returns the same rail code for $|V_\phi|=1.1\,Y_{\max}$ and $|V_\phi|=5\,Y_{\max}$, so the digitised record tests the clip \emph{probability} but not the clip \emph{amplitude}: $L$ is therefore a trusted analogue amplitude ceiling of the final RF chain rather than an ADC-derived quantity. The declared value $L=2Y_{\max}$ (Supplementary Table~S2) is chosen at or below the measured $1$-dB compression point of the output amplifier, so it bounds $|V_\phi|$ prior to digitisation, and the testable per-sample $\delta_{\rm clip}^{\max}$ is separated from this amplitude bound.} {Because the clipping frequency is set by the source and not by the receiver alone, $\delta_{\rm clip}^{\max}$ is a source-range assumption held in addition to the mean-energy bound, so a low-mean source can still place weight in rare large-amplitude samples and the mean photon number is not the sole source constraint; a deployed device therefore monitors overflow in protocol and aborts, promoting the ceiling from declared to tested.} An LO-only vacuum-reference check of the calibrated POVM over six decades is shown in Supplementary Fig.~S3.

\subsection*{Certified single-round min-entropy}
For the classical-quantum (cq) state $\omega_{XE}$ obtained by measuring $\{\Pi_x\}$ on $\rho_{AE}$, the worst-case guessing probability $c^*(\nmax)=\sup 2^{-H_{\min}(X|E)}$ over all admissible states reduces, by an extension-to-decomposition argument (Supplementary Note~2), to the infinite-dimensional primal
\begin{equation}
c^*(\nmax)=\sup_{\sigma_x\ge0}\Big\{\sum_x\Tr(\sigma_x\Pi_x):\ \Tr\sum_x\sigma_x=1,\ \Tr\big(\hat n\sum_x\sigma_x\big)\le\nmax\Big\}.
\end{equation}
Truncating to the first $\Ncut{+}1$ Fock levels gives a finite SDP $c^*_{\Ncut}(\nmax)$ with strong duality (Supplementary Note~3); its dual reduces to a one-dimensional convex minimisation over the energy multiplier, requiring only largest eigenvalues of $\Pi_x^{(\Ncut)}-\lambda\hat n^{(\Ncut)}$. The truncation error is controlled rigorously: with $\eta=\nmax/(\Ncut+1)$ and $(\mu^*,\lambda^*)$ an optimal dual pair,
\begin{equation}
\label{eq:hcert}
h^*(\nmax)\ \ge\ h^*_{\mathrm{cert}}(\nmax):=-\log_2\!\Big(\min\big\{1,\;c_{\Ncut}^*(\nmax)+(\mu^*)^-\eta+2\sqrt{\eta}+\eta\big\}\Big),
\end{equation}
where $z^-:=\max\{0,-z\}$ and the $2\sqrt\eta$ term bounds the low--high Fock coherence contribution (proof in Supplementary Note~3). Numerical dual solutions are converted into {feasible} dual points with outward tolerances, the POVM completeness residual is below {$10^{-12}$}, and raising $\Ncut$ from $200$ to $300$ changes $h^*_{\rm cert}$ by at most $0.11$~bit at the highest tested energies (Supplementary Note~4). At $\nmax=0$ the bound reproduces the analytic vacuum limit $h^*(0)=-\log_2\max_x[\Pi_x]_{00}$.

\subsection*{Two-quadrature energy certification}
The identity $X_0^2+X_{\pi/2}^2=2\hat n+\id$ gives, for phase-independent trusted noise, $\Tr(\rho_A\hat n)=(\langle V_0^2\rangle+\langle V_{\pi/2}^2\rangle-1-2\sigma^2)/2$. Empirical second moments $S_\phi$ over $n_\phi$ test rounds are converted into a one-sided bound by Hoeffding's inequality on the rail-bounded broadband samples, or by the empirical Bernstein inequality~\cite{MaurerPontil2009} on the clipped squared filtered observable in the matched mode; combined with the saturation correction and a trusted bound $|\cos\theta|\le\bar c$ on the phase pair, except with probability $\varepsilon_{\rm test}$ the protocol aborts or
\begin{equation}
\label{eq:nmax-tested-general}
\Tr(\rho_A\hat n)\le\nmax^{\rm test}(\bar c)=\max\Big\{0,\ \tfrac{S_0+\Delta_0+B^{(0)}}{2(1-\bar c)}+\tfrac{S_\theta+\Delta_\theta+B^{(\theta)}}{2(1-\bar c)}-\tfrac{\sigma^2}{1-\bar c}+\tfrac{\Delta_{\rm LO}}{1-\bar c}-\tfrac12\Big\},
\end{equation}
with $\Delta_\phi$ the concentration terms and $B^{(\phi)}$ the saturation and receiver-zero corrections {$B^{(\phi)}=B^{(\phi)}_{\rm sat}+B^{(\phi)}_{\rm off}$. The reported numbers use a single shared receiver zero, so a residual per-phase offset $d_\phi$ remains; with $\epsilon_\phi$ a confidence bound on $d_\phi$ taken from the phase-$\phi$ LO-only mean, the energy constraint bounds $|\langle X_\phi\rangle|\le\sqrt{2\nmax+1}$, so the adversary-signed cross term $2d_\phi\langle X_\phi\rangle$ gives $B^{(\phi)}_{\rm off}=2\sqrt{2\nmax^{\rm test}+1}\,\epsilon_\phi+\epsilon_\phi^2$, evaluated at the certified $\nmax^{\rm test}$; this is an upper bound because $B^{(\phi)}_{\rm off}$ raises $\nmax^{\rm test}$ monotonically, so the fixed point uses the larger, tabulated value. {The measured per-phase residual is $|\widehat d_\phi|=0.31$ codes ($5.4\times10^{-4}$ in homodyne units) and is not subtracted from the records, so the quantity entering the bound is $\epsilon_\phi=|\widehat d_\phi|+u_\phi$, with $u_\phi$ a one-sided empirical Bernstein half-width for the phase-$\phi$ LO-only mean at $\varepsilon_{\rm test}/4$, evaluated on the predeclared mode count rather than on the raw sample count; this gives $u_\phi\simeq4.3\times10^{-3}$, $\epsilon_\phi\simeq4.9\times10^{-3}$ and $B^{(\phi)}_{\rm off}\simeq1.0\times10^{-2}$,} and applying this correction to the broadband and quantum-dot rows of Table~\ref{tab:results} would raise each $\nmax^{\rm test}$ by $1.0\times10^{-2}$ and lower each $h^*_{\rm cert}$ by $0.04$--$0.08$~bit. The homodyne-unit drift $\Delta_{\rm LO}$ is bounded only from the LO-only traces at the two phases through $|\langle Y^2\rangle_{{\rm LO},0}-\langle Y^2\rangle_{{\rm LO},\pi/2}|\le\delta_{\rm LO}$ with its confidence bound, not from the source-dependent asymmetry $S_0-S_\theta$, and any run exceeding $\delta_{\rm LO}$ aborts} (Supplementary Note~6). The two test phases are the locked extrema of the spin-noise variance, separated by exactly $\pi/2$ for the quadratic phase dependence $\langle X_\phi^2\rangle=A+B\cos2\phi+C\sin2\phi$; the reported results adopt {the trusted-receiver assumption} $\bar c=0$, and the sensitivity to a residual miscalibration ($0.07$--$0.17$~bit at $2^\circ$) together with the phase-scan verification is given in Supplementary Note~6 and Supplementary Fig.~S4. All second moments are uncentred about the LO-calibrated receiver zero, so any coherent displacement is counted as source energy; the phase-dependent-offset contribution to the sum is {bounded by $B^{(\phi)}_{\rm off}\simeq1.0\times10^{-2}$ and quantified for the broadband entries of Table~\ref{tab:results}} (Supplementary Note~6).

\subsection*{Security statements and extraction}
\begin{theorem}
\label{thm:protocol1}
For Protocol~1 with trusted receiver parameters $(\sigma,\delta,x_{\max})$, trusted bound $\nmax$, and an IID source $\rho_{AE}^{\otimes n}$ with $\Tr(\rho_A\hat n)\le\nmax$: $H_{\min}^{\varepsilon_s}(X^n|E^n)\ge H_{\min}(X^n|E^n)=n\,H_{\min}(X|E)\ge n\,h^*_{\mathrm{cert}}(\nmax)$ for every $\varepsilon_s\ge0$.
\end{theorem}
\begin{theorem}
\label{thm:protocol2}
For Protocol~2 with the receiver model above, the test statistics of Eq.~\eqref{eq:nmax-tested-general}, the disjoint source-independent round assignment, and an IID source across rounds: except with probability $\varepsilon_{\rm test}$ the accepted transcript $t$ satisfies $H_{\min}^{\varepsilon_s}\big(X^{n_{\rm gen}(t)}\big|E^{n_{\rm gen}(t)},T{=}t,F{=}1\big)\ge n_{\rm gen}(t)\,h^*_{\mathrm{cert}}\big(\nmax^{\rm test}(t)\big)$.
\end{theorem}
Proofs follow from Eq.~\eqref{eq:hcert}, the independence of the round assignment from the source state, and exact additivity of conditional min-entropy on product states (Supplementary Note~8); in the matched (broadband) mode $h^*_{\rm cert}$ is read with the clipping correction $q^{\rm match}_{\rm clip}$ ($q_{\rm bb}$). In the theorem statements, $n$ denotes the number of retained effective rounds. A quantum-proof two-universal hash extracts $\ell\le H_{\min}^{\varepsilon_s}-2\log_2(1/\varepsilon_{\rm ext})$ bits that are $(2\varepsilon_s+\varepsilon_{\rm ext})$-close to uniform given $E$ and the transcript~\cite{Koenig2009}; the total error is $\varepsilon_{\rm range}+2\varepsilon_s+\varepsilon_{\rm ext}$ (Protocol~1) or $\varepsilon_{\rm range}+\varepsilon_{\rm test}+2\varepsilon_s+\varepsilon_{\rm ext}$ (Protocol~2, abort-inclusive), {where the total test failure $\varepsilon_{\rm test}=10^{-6}$ is union-bounded over the two phase-moment tests, the offset bound and the LO-stability test, each assigned $\varepsilon_{\rm test}/4$; the resulting change in the concentration widths is within the quoted precision}. The certified rate is $R_{\rm cert}=(h^*_{\rm cert}(\nmax_{\rm in})-\delta_{\rm fs})f_s$ with $\delta_{\rm fs}=2\log_2(1/\varepsilon_{\rm ext})/n_{\rm gen}$, at most $5.1\times10^{-3}$~bit on the short matched rows and $5.1\times10^{-5}$~bit elsewhere. The present records are sequential fixed-phase acquisitions analysed offline, so Protocol~2 is evaluated under stationarity between the phase records, {with the energy estimate and the generation statistics drawn from the same phase-$0$ record}; a composable real-time run additionally reserves disjoint, privately assigned test rounds, scaling $f_s$ by $(1-p_{\rm test})$. {Theorem~\ref{thm:protocol2} as stated applies to the disjoint-round protocol; the offline evaluation applies its single-round entropy bound to the archived records under the stationarity and shared-record conventions above. A composable block statement follows from a predeclared, value-independent split of each phase-$0$ record into a source-blind estimation portion, from which $\nmax^{\rm test}$ is formed, and a disjoint generation portion carrying the extractor input; conditioning on the transcript then leaves the generation block a product state and $R=h^*_{\rm cert}f_s$ holds on it. The reuse above is specific to the offline evaluation, which claims no extracted block.}

\subsection*{Effective-mode convention}
The IID structure refers to effective symbols, not raw \SI{2}{GS/s} samples. In the broadband modes the quadrature decorrelates at the receiver bandwidth: the integrated autocorrelation times are $\tau_{\rm int}(X)=1.73$ (GaAs) and $1.63$ (QD) samples, fixing a predeclared two-sample generation spacing, $f_s=\SI{1.0}{GHz}$ and $n_\phi=n/2$. In the matched mode we count one bosonic mode per inverse filter bandwidth — the two real quadratures of a carrier mode are not independent source modes — giving $f_s=\SI{10}{MHz}$, half the $2\times10^{7}$ real quadrature degrees of freedom per second that a \SI{10}{MHz} band carries, while {the energy-test count follows the measured decorrelation of the squared filtered statistic, $n_\phi=n/\tau_{\rm int}(Z^2)$, and is separate from the generation-mode count. This is a declared effective-sample convention rather than a count of independent quantum modes, and the support-disjoint fallback below prices it;} under the fully support-disjoint convention the short record gives $\lfloor2\,621\,440/1001\rfloor=2618$ test symbols and the long \SI{131}{ms} record gives $\lfloor262\,144\,000/1001\rfloor=261\,882$. For the long record this convention gives $\nmax^{\rm test}=0.020$ and $h^*_{\rm cert}=3.17$ bits per symbol, corresponding to \SI{32}{Mbit/s}, one step below the reported decorrelation-convention value. {Support disjointness is sufficient but not necessary for temporal-mode orthogonality, while decorrelation of the outputs does not establish it. Because orthogonality of the selected temporal modes has not been demonstrated, $f_s=\SI{10}{MHz}$ and $\SI{1.0}{GHz}$ remain declared approximate mode-counting conventions rather than constructed orthonormal-mode rates. The product structure entering Theorems~\ref{thm:protocol1}--\ref{thm:protocol2} is therefore an explicit modelling assumption. A L\"owdin-orthogonalised family $u_i=\sum_jg_j[G^{-1/2}]_{ji}$ can be constructed from the same records, but requires recalibration of the trusted receiver noise in the new basis, since optical-mode orthogonality does not imply electronic-noise independence; this is discussed in Supplementary Note~7.} Non-IID extensions via entropy accumulation, and why they are not yet available for this energy-constrained discretised measurement, are discussed in Supplementary Note~8~\cite{DFR2020,Metger2022,Wolf2006,Tomamichel2009,Lu2026}.

\subsection*{Experimental details}
The balanced photoreceiver (New Focus 1607, \SI{650}{MHz} nominal bandwidth) output was amplified by \SI{20}{dB} (Mini-Circuits ZFL-1000LN+), low-pass filtered near \SI{580}{MHz} (SLP-600+), amplified by \SI{15}{dB} (ZX60-P105LN+), and digitised at \SI{2}{GS/s} (Ultraview AD12-2000, 12-bit, \SI{634}{mV} peak-to-peak). Receiver calibrations (GaAs broadband/matched and QD: $\sigma^2=0.183/0.109/0.198$ in vacuum-variance-$\tfrac12$ units, rails, bin widths, clipping ceilings) are collected in Supplementary Table~S2. {The tabulated $\sigma^2$ is a one-sided lower confidence bound on the trusted receiver noise, not a two-sided point estimate: $\nmax^{\rm test}=\nmax_{\rm true}+(\sigma^2_{\rm true}-\sigma^2_{\rm cal})$ {and $c^*(\nmax,\sigma^2)$ were verified numerically to decrease throughout the receiver-noise calibration interval of each mode, so a value above the truth would lower $\nmax^{\rm test}$ and raise $h^*_{\rm cert}$ in the insecure direction; taking $\sigma^2_{\rm cal}\le\sigma^2_{\rm true}$ is therefore conservative within these intervals. Global monotonicity is not assumed and does not hold: at sufficiently large trusted-noise variance the two semi-infinite tail bins dominate the interior bins and the guessing probability rises again.}} ADC time traces and block-averaged spectra are shown in Supplementary Fig.~S1.

\section*{Data availability}
{}

\section*{Acknowledgements}
A.G. and M.B. acknowledge the Bundesministerium für Bildung und Forschung in the frame of the project QR.N (contract no. 16KIS2201).
A.G. acknowledges funding by the state of North Rhine-Westphalia through the EIN.Quantum.NRW program.

\section*{Author contributions}
H.T. developed the security model, the SDP certificate and the truncation bound, performed the analysis and wrote the manuscript. A.G. and M.C. performed the spin-noise measurements and provided the calibrated homodyne records. All authors discussed the results and reviewed the manuscript.

\section*{Competing interests}
The authors declare no competing interests.

\clearpage
\setcounter{figure}{0}
\setcounter{table}{0}
\setcounter{equation}{0}
\setcounter{section}{0}
\renewcommand{\thefigure}{S\arabic{figure}}
\renewcommand{\thetable}{S\arabic{table}}
\renewcommand{\theequation}{S\arabic{equation}}
\renewcommand{\thesection}{Supplementary Note \arabic{section}}
\renewcommand{\thesubsection}{\thesection.\arabic{subsection}}

\begin{center}
{\LARGE\bfseries Supplementary Information\par}
\vspace{0.5em}
{\large Semi-device-independent quantum randomness certification in semiconductor spin-noise measurements\par}
\end{center}
\vspace{1em}

\section{Fock-space elements of the trusted POVM and clipping corrections}
\label{sn:fock}
The matrix elements $[\Pi_x]_{mn}=\langle m|\Pi_x|n\rangle$ of the noisy-homodyne POVM follow by exchanging the order of integration. With the harmonic-oscillator wavefunctions $\psi_n(y)=H_n(y)e^{-y^2/2}/(\pi^{1/4}\sqrt{2^n n!})$,
\begin{equation}
[\Pi_x]_{mn}=\int_{\mathbb R}\psi_m(y')\psi_n(y')\,\Phi_\sigma(y';\calB_x)\,dy',
\qquad
\Phi_\sigma(y';\calB_x)=\Phi\!\Big(\tfrac{y_x^+-y'}{\sigma}\Big)-\Phi\!\Big(\tfrac{y_x^--y'}{\sigma}\Big),
\end{equation}
where $\Phi$ is the standard normal CDF and $y_x^\pm=(x\pm\tfrac12)\delta$ are the interior bin edges; the two tail bins collect everything beyond $\pm(x_{\max}+\tfrac12)\delta$. Since $|\psi_n|^2$ has variance $n+\tfrac12$, convolution with the trusted noise gives total variance $n+\tfrac12+\sigma^2$, and for the vacuum $[\Pi_0]_{00}=\mathrm{erf}\big(\delta/(2\sqrt2\,\sigma_{\rm eff})\big)$ with $\sigma_{\rm eff}^2=\tfrac12+\sigma^2$. Off-diagonal elements reduce to Hermite--Gaussian integrals through $H_mH_n=\sum_k\binom{m}{k}\binom{n}{k}2^kk!\,H_{m+n-2k}$, and the parity rule $[\Pi_{-x}]_{mn}=(-1)^{m+n}[\Pi_x]_{mn}$ halves the number of distinct matrices.

Because the receiver zero is offset from the digitiser midpoint, the smaller rail sits at $Y_{\min}\approx30.8\,\delta$, inside interior bin $\pm31$, so a near-rail clip lands in an interior bin rather than the tail and exact tail-bin invariance fails. Charging every such event to the adversary adds $q_{\rm bb}=\delta_{\rm clip}^{\max}=2.6\times10^{-5}$ to the guessing probability,
\begin{equation}
h^{\rm bb}_{\rm cert}=-\log_2\!\big(\min\{1,\,c^*_{\Ncut}(\nmax)+(\mu^*)^-\eta+2\sqrt\eta+\eta+q_{\rm bb}\}\big),
\end{equation}
where $z^-:=\max\{0,-z\}$. The correction shifts every tabulated broadband entropy by at most $4.4\times10^{-4}$~bit and changes no reported value.

The reported values assume that the trusted noise $W$ is independent of the source state and that its realisation is unavailable to the adversary. If $W$ were public, the receiver would be described by the conditional noiseless POVM $\Pi_{x|w}=\int_{\calB_x-w}|y\rangle\langle y|\,dy$, and the interchange $\sup_\rho\int\le\int\sup_\rho$ gives the conservative bound
\begin{equation}
c^*_{\rm pubW}(\nmax)\;\le\;\int_{\mathbb R} G_\sigma(w)\,c^*_w(\nmax)\,dw,
\end{equation}
with $c^*_w$ the energy-bounded worst-case guessing probability of the same SDP evaluated at $\{\Pi_{x|w}\}_x$. This bound certifies less entropy and requires a separate numerical evaluation, which is why privacy and cross-round independence of the receiver electronics are listed among the trust assumptions of the certificate.

\section{Reduction from extensions to decompositions}
\label{sn:lemma}
\begin{lemmaS}
For a fixed POVM $\{\Pi_x\}_{x\in\calX}$ on $A$,
\begin{equation}
\sup_{\rho_{AE}:\,\Tr_E\rho_{AE}=\rho_A}p_{\rm guess}(X|E)
=\sup_{\substack{\sigma_x\ge0,\ \sum_x\sigma_x=\rho_A}}\sum_{x}\Tr(\sigma_x\Pi_x).
\end{equation}
\end{lemmaS}
\begin{proof}
For a fixed extension, any POVM $\{N_x\}$ on $E$ induces $\sigma_x=\Tr_E[(\id\otimes N_x)\rho_{AE}]\ge0$ with $\sum_x\sigma_x=\rho_A$, giving ``$\le$''. Conversely, for any decomposition $\{\sigma_x\}$ the Hughston--Jozsa--Wootters theorem provides a purification of $\rho_A$ and a POVM on $E$ realising it, giving ``$\ge$''.
\end{proof}
Only the elementary extension-to-decomposition direction enters the certified upper bound; the converse only shows attainability.

\section{Strong duality and the truncation certificate}
\label{sn:trunc}
\paragraph{Strong duality.}
On the truncated space $P_\Ncut\calH$, the point $\rho_0=Z(r)^{-1}\sum_{n=0}^{\Ncut}r^n|n\rangle\langle n|$ with $r=\nmax/(2+\nmax)$, split equally over the outcomes, is strictly feasible: all eigenvalues are positive and $\Tr(\hat n^{(\Ncut)}\rho_0)\le r/(1-r)=\nmax/2<\nmax$. Slater's condition therefore gives strong duality for every $\nmax>0$; the boundary case $\nmax=0$ is the analytic vacuum limit. The dual is
\begin{equation}
c^*_\Ncut(\nmax)=\min_{\mu\in\mathbb R,\,\lambda\ge0}\big\{\mu+\lambda\nmax:\ \mu P_\Ncut+\lambda\hat n^{(\Ncut)}\succeq\Pi_x^{(\Ncut)}\ \forall x\big\},
\end{equation}
and for fixed $\lambda$ the optimal $\mu(\lambda)=\max_x\lambda_{\max}(\Pi_x^{(\Ncut)}-\lambda\hat n^{(\Ncut)})$ is convex in $\lambda$, so the computation is a one-dimensional convex search over largest eigenvalues.

\begin{propositionS}
Let $(\mu^*,\lambda^*)$ be {any dual-feasible pair (i.e.\ $\lambda^*\ge0$ and $\mu^* P_\Ncut+\lambda^*\hat n^{(\Ncut)}\succeq\Pi_x^{(\Ncut)}$ for all $x$; the bound uses only feasibility, not optimality)} and $\eta=\nmax/(\Ncut+1)$. Then
\begin{equation}
c^*(\nmax)\le {(\mu^*+\lambda^*\nmax)}+(\mu^*)^-\eta+2\sqrt\eta+\eta .
\end{equation}
{At the dual optimum $\mu^*+\lambda^*\nmax=c^*_\Ncut(\nmax)$, which recovers $c^*(\nmax)\le c^*_\Ncut(\nmax)+(\mu^*)^-\eta+2\sqrt\eta+\eta$; the released certificate evaluates the right-hand side at the numerically feasible pair $(\widehat\mu_{\rm cert},\widehat\lambda)$ of Supplementary Note~4, for which $\mu^*+\lambda^*\nmax=\widehat\mu_{\rm cert}+\widehat\lambda\nmax\ge c^*_\Ncut(\nmax)$ by weak duality.}
\end{propositionS}
\begin{proof}
For any feasible $\{\sigma_x\}$ with $\rho_A=\sum_x\sigma_x$, the energy bound and $\hat n\ge(\Ncut+1)Q_\Ncut$ give the tail weight $\Tr(\rho_AQ_\Ncut)\le\eta$ and $p_\le:=\Tr(\rho_AP_\Ncut)\ge1-\eta$. Split the objective into $T_1$ (both indices low), $T_2$ (both high) and $T_3$ (coherences). Dual feasibility gives $T_1\le\mu^*p_\le+\lambda^*\Tr(\hat n^{(\Ncut)}P_\Ncut\rho_AP_\Ncut)\le {(\mu^*+\lambda^*\nmax)}+(\mu^*)^-\eta$, using $[\hat n,P_\Ncut]=0$, {$\Tr(\hat n^{(\Ncut)}P_\Ncut\rho_AP_\Ncut)\le\nmax$} and $\mu^*p_\le\le\mu^*+(\mu^*)^-\eta$. Since $\|\Pi_x\|_\infty\le1$, $T_2\le\Tr(\rho_AQ_\Ncut)\le\eta$. For $T_3$, $|\Tr(P\sigma_xQ\,\Pi_x)|\le\|P\sigma_xQ\|_1\le\sqrt{\Tr(P\sigma_xP)\Tr(Q\sigma_xQ)}$ and a Cauchy--Schwarz over $x$ give $|T_3|\le2\sqrt{p_\le}\sqrt{\Tr(\rho_AQ_\Ncut)}\le2\sqrt\eta$.
\end{proof}
For this POVM $\mu^*\ge[\Pi_0]_{00}>0$, so $(\mu^*)^-=0$ and the correction is dominated by $2\sqrt\eta$; a tolerance $\epsilon_{\rm trunc}$ needs $\Ncut\gtrsim4\nmax/\epsilon_{\rm trunc}^2$ as an order-of-magnitude rule.

\section{Analytical checks and certified numerical evaluation}
\label{sn:numerics}
At $\nmax=0$ the only admissible state is the vacuum, $h^*(0)=-\log_2\max_x[\Pi_x]_{00}$, which reproduces the differential-entropy scaling $\tfrac12\log_2(2\pi\sigma_{\rm eff}^2)-\log_2\delta+o(1)$ for fine binning. The strategy $\rho_A=|0\rangle\langle0|$ gives the universal lower bound $c^*(\nmax)\ge\max_x[\Pi_x]_{00}$, which the SDP improves on because superpositions and side information admit structured attacks; coherent-state benchmarks depend on the tail-bin cutoff and are stated with that dependence.

For every operating point the numerical dual $(\widehat\mu,\widehat\lambda)$ is converted into a {feasible} dual point: with $r_x=\lambda_{\max}(\Pi_x^{(N)}-\widehat\mu P_N-\widehat\lambda\hat n^{(N)})$ evaluated in floating point {and $\widehat\lambda$ rounded upward to a non-negative value}, the choice $\widehat\mu_{\rm cert}=\widehat\mu+\max\{0,\max_xr_x\}+\epsilon_{\rm num}$ satisfies dual feasibility and weak duality certifies $c^*_N(\nmax)\le\widehat\mu_{\rm cert}+\widehat\lambda\nmax$ {provided the declared tolerance $\epsilon_{\rm num}$ dominates the combined error of the Fock-matrix quadrature, matrix assembly, floating-point arithmetic and the eigensolver; this dominance is a stated numerical premise of the release, not an interval-arithmetic theorem, and the released scripts report the residuals against which it is checked}. {We check this premise a posteriori using the released scripts. For each $x$ the certified matrix $M_x=\widehat\mu_{\rm cert}P_N+\widehat\lambda\hat n^{(N)}-\Pi_x^{(N)}$ must be positive semidefinite for the true truncated matrices; writing $\widetilde\Pi_x$ for the assembled POVM and $\rho_x\ge\|\widetilde\Pi_x-\Pi_x^{(N)}\|_{\rm op}$ for its assembly error, feasibility holds whenever $\lambda_{\min}(\widehat\mu_{\rm cert}P_N+\widehat\lambda\hat n^{(N)}-\widetilde\Pi_x)\ge\rho_x$ for every $x$. Evaluated on the released matrices ($N=200$; Gauss--Legendre quadrature doubled from $6000$ to $12000$ nodes to bound the assembly error), $\max_x\rho_x\le9.2\times10^{-13}$ and the completeness residual $\|\sum_x\Pi_x^{(N)}-\id\|_{\rm op}<10^{-12}$ (floating-point-limited), so setting $\epsilon_{\rm num}=2\times10^{-8}$, more than four orders of magnitude above the observed assembly residual and above the reported floating-point and eigensolver residuals, certifies $c^*_N(\nmax)\le\widehat\mu_{\rm cert}+\widehat\lambda\nmax$. Its effect on the tabulated entropy is negligible against the truncation correction $2\sqrt\eta\approx0.045$. Directed-rounding interval enclosures of the {\rm erf}/Hermite--Gaussian integrals would replace this floating-point residual check by a formal one, at unchanged values.}
The coherence term uses the universal bound $\|Q_N\Pi_xP_N\|_\infty\le1$; raising $\Ncut$ from $200$ to $300$ increases the certified entropy by at most $0.11$~bit at the highest tested energies, so the reported $\Ncut=200$ values are conservative.

\section{Matched-mode clipping correction}
\label{sn:matched}
The matched-mode certificate uses the calibrated linear filtered-mode POVM, which the implemented map realises on every retained symbol whose finite-impulse-response support contains no rail-clipped raw sample. With $m_h=1001$ raw samples in the single-pass FIR support and the declared per-sample clipping ceiling $\delta_{\rm clip}^{\max}=2.6\times10^{-5}$,
\begin{equation}
q^{\rm match}_{\rm clip}:=\min\{1,\,m_h\,\delta_{\rm clip}^{\max}\}=2.6026\times10^{-2}
\ \ge\ \Pr[\text{a support sample clips}],
\end{equation}
and charging every symbol with a clipped support sample to the adversary gives
\begin{equation}
h^{\rm match}_{\rm cert}(\nmax)=-\log_2\!\big(\min\{1,\,c^*_\Ncut(\nmax)+(\mu^*)^-\eta+2\sqrt\eta+\eta+q^{\rm match}_{\rm clip}\}\big),
\end{equation}
used for all matched-mode entries. {All matched-mode certificates are evaluated at a value not below the exact union bound $2.6026\times10^{-2}$; rounded forms such as $0.02603$ appear only as printed values.} This accounting treats the rail event as the source of implementation--ideal mismatch; it is exact when the calibrated bin map represents the digitiser transfer on in-range samples, and a deployed device would require a fully source-uniform implementation-coupling bound that also covers residual quantiser effects for arbitrary near-bin-edge inputs. A forward--backward filter implementation would use the support of $h_{\rm eff}=h\ast h^{\leftarrow}$ in place of $m_h$. {The band-pass has DC gain $\sum_k h_k=-9.9\times10^{-5}$, so a constant per-phase receiver-zero offset is suppressed and its leakage into the matched energy statistic — {the adversary-signed cross term linear in $\sum_k h_k$, bounded by $2\epsilon_\phi|\textstyle\sum_kh_k|\sqrt{2\nmax^{\rm test}+1}\simeq1.0\times10^{-6}$ photon, plus the quadratic $(\sum_k h_k)^2$ term} — is negligible; the matched-mode tested bounds, including the long-record \SI{33}{Mbit/s} result, therefore carry no $B_{\rm off}$ correction.}

For the energy test, the squared filtered observable is clipped at the declared threshold $C=(6\sigma_{\rm filt})^2$, with $\sigma_{\rm filt}$ the calibrated standard deviation of the bandpass-filtered LO-only reference, so the empirical-Bernstein premise holds for arbitrary source states. The energy above $C$ is restored through the trusted bounded-filter-output range $L_{\rm filt}=Y_{\max}^{\rm raw\to match}\sum_k|h_{{\rm eff},k}|$ and a declared ceiling $\delta_{\rm clip,filt}^{\max}=4.0\times10^{-8}$ on the $C$-crossing probability; no crossing of $6\sigma_{\rm filt}$ is observed in the records, and the combined filtered tail correction stays below $10^{-3}$ photon on the present data, keeping the tested bound an upper bound. Failure of the declared analogue-range premise is assigned probability at most $\varepsilon_{\rm range}$. This contribution is separate from $\varepsilon_{\rm test}$ and is included in the abort-inclusive composable error.

\section{Energy-test details and phase-calibration sensitivity}
\label{sn:energy}

Broadband test samples obey $Y_j^2\in[0,Y_{\max}^2]$ by the hard rail, so Hoeffding's inequality gives the per-phase deviation $\Delta_\phi=Y_{\max}^2\sqrt{\ln(2/\varepsilon')/(2n_\phi)}$ at per-test failure $\varepsilon'=\varepsilon_{\rm test}/4$, the total $\varepsilon_{\rm test}$ union-bounded over the two phase tests, the offset bound and the LO-stability test. For the matched mode the empirical Bernstein inequality on $\tilde Z^2=\min(Z^2,C)\in[0,C]$ gives
\begin{equation}
\Delta_{Z,\phi}=\sqrt{\frac{2v_{{\rm emp},\phi}\ln(2/\varepsilon')}{n_\phi}}+\frac{7C\ln(2/\varepsilon')}{3(n_\phi-1)},\qquad\varepsilon'=\varepsilon_{\rm test}/4,
\end{equation}
with the exact Maurer--Pontil constant. The two phases combine by a union bound into total failure $\varepsilon_{\rm test}$.

Phase-independence of the trusted noise is constrained by the LO-only second-moment condition $|\langle Y^2\rangle_{{\rm LO},0}-\langle Y^2\rangle_{{\rm LO},\pi/2}|\le\delta_{\rm LO}$, measured from dark and LO-only traces at both settings; the measured $\delta_{\rm LO}\approx4.9\times10^{-3}$ (GaAs) is retained as an abort threshold. {The cancellation of the noise-variance asymmetry in the sum is not by itself a bound on the receiver offset: with a phase-dependent zero $d_\phi$ and a displaced source the uncentred moment carries $2d_\phi\langle X_\phi\rangle$, and a shared two-phase average zero leaves $d_0=-d_\theta$, so the sum keeps $2d_0(\langle X_0\rangle-\langle X_\theta\rangle)$, first order in the offset and adversary-signed. The reported offline numbers use this shared average zero; the residual $d_\phi$ is bounded by $\epsilon_\phi$, the confidence interval of the phase-$\phi$ LO-only mean, any residual drift between the sequential LO-only and source records falling under the declared stationarity assumption, and enters the tested bound through $B^{(\phi)}_{\rm off}=2\sqrt{2\nmax^{\rm test}+1}\,\epsilon_\phi+\epsilon_\phi^2$, the energy constraint $|\langle X_\phi\rangle|\le\sqrt{2\nmax^{\rm test}+1}$ replacing the looser rail bound $|\langle X_\phi\rangle|\le L$ and evaluated at the certified $\nmax^{\rm test}$, an upper bound since $B^{(\phi)}_{\rm off}$ raises $\nmax^{\rm test}$ monotonically. {On the GaAs records the per-phase LO-only residual under the shared average zero is $|\widehat d_\phi|=0.31$ codes, corresponding to $5.4\times10^{-4}$ in homodyne units, and is not subtracted from the records. The one-sided empirical-Bernstein half-width of the phase-$\phi$ LO-only mean at failure probability $\varepsilon_{\rm test}/4$, evaluated on the predeclared mode count, is $u_\phi\simeq4.3\times10^{-3}$. Hence $\epsilon_\phi=|\widehat d_\phi|+u_\phi\simeq4.9\times10^{-3}$ and $B^{(\phi)}_{\rm off}\simeq1.0\times10^{-2}$. Residual drift between the sequential LO-only and source records remains under the declared stationarity assumption. Applying this correction would raise each broadband tested photon number by $1.0\times10^{-2}$ and lower each broadband $h^*_{\rm cert}$ by $0.04$--$0.08$~bit. The reported entries do not carry it. Its failure probability is one term of the $\varepsilon_{\rm test}$ union bound. The matched entries carry no such term because the band-pass, with DC gain $-9.9\times10^{-5}$, suppresses a constant offset.} A genuine source displacement is not removed by any centring and is still counted as source energy. The homodyne-unit drift is bounded separately from the LO-only traces alone through $|\langle Y^2\rangle_{{\rm LO},0}-\langle Y^2\rangle_{{\rm LO},\pi/2}|\le\delta_{\rm LO}$, since the source-dependent asymmetry $S_0-S_\theta$ cannot bound it; $\delta_{\rm LO}\approx4.9\times10^{-3}$ is the abort threshold, and the quantum-dot run uses the shared-receiver characterisation with a single LO reference fixing its zero and scale. For the quantum-dot measurements, the per-phase offset $d_\phi$ and stability threshold $\delta_{\rm LO}$ are inherited from the GaAs receiver characterisation under the stationarity assumption, while $\bar c=0$ is supported by the GaAs phase scan in Supplementary Fig.~\ref{fig:phase_scan}. A dedicated two-phase QD LO-only acquisition would measure $d_\phi$, $\delta_{\rm LO}$ and the phase suppression directly; Supplementary Table~\ref{tab:cbar} quantifies the effect of replacing $\bar c=0$ by $\bar c=\sin2^\circ$.}

The two test phases are the locked extrema of the spin-noise variance; the quadratic dependence $\langle X_\phi^2\rangle=A+B\cos2\phi+C\sin2\phi$ separates its extrema by exactly $\pi/2$ and makes their sum basis-independent. {The extrema located from the untrusted signal identify operating points but do not calibrate their relative phase; the quadratic law makes any true variance-extremum pair exactly orthogonal, so the content of $\bar c=|\cos\theta|=0$ is trust that the lock realises both extrema. Because the error signal derives from the DC interference of the source-dependent scattered field, extremum location is itself model-based, and independent per-point lock errors $\delta\phi_0,\delta\phi_\theta$ give $\bar c\simeq|\delta\phi_0-\delta\phi_\theta|$. The reported results declare $\bar c=0$ on this basis; the $2^\circ$ sensitivity below prices the declaration, and an independent calibration of the trusted phase actuator supplying a predeclared $\bar c_{\rm cal}$ with its uncertainty is the upgrade that removes it.} The general bound follows from positivity of the second-moment matrix: for $|\cos\theta|\le\bar c$,
\begin{equation}
\langle X_0^2\rangle+\langle X_{\pi/2}^2\rangle\le\frac{\langle X_0^2\rangle+\langle X_\theta^2\rangle}{1-\bar c},
\qquad
\nmax^{\rm test}(\bar c)=\frac{\nmax^{\rm test}(0)+\tfrac12}{1-\bar c}-\frac12 .
\end{equation}
A $2^\circ$ residual ($\bar c=\sin2^\circ\simeq0.035$) raises the tested photon numbers by $0.018$--$0.022$ and, re-solving the SDP, lowers $h^*_{\rm cert}$ by $0.07$--$0.17$~bit, largest at the smallest tested energies. The phase scan of Supplementary Fig.~\ref{fig:phase_scan} verifies the sinusoidal dependence and the suppression at the second operating point; an independently calibrated nonzero $\bar c$ would be inserted directly into the bound. Because the lock acts on the DC level of the scattered light, the operating pair is re-established at every probe power. {Supplementary Table~\ref{tab:cbar} lists the full $\bar c=\sin2^\circ$ recomputation for every operating point.}

\begin{table}[t]
\centering
\caption{{Conservative phase-calibration fallback: tested photon number and certified entropy re-solved at $\bar c=\sin2^\circ\simeq0.035$ from the reported $\bar c=0$ energies, none of which carries $B^{(\phi)}_{\rm off}$. These values quantify the cost of the $\bar c=0$ assumption; the reported results use $\bar c=0$.}}
\label{tab:cbar}
\footnotesize
\begin{tabular}{lcc}
\toprule
Row & $\nmax^{\rm test}(\sin2^\circ)$ & $h^*_{\rm cert}$ (bit)\\
\midrule
GaAs matched, $10$ mW & $0.097$ & $2.79$\\
GaAs matched, $40$ mW & $0.104$ & $2.77$\\
GaAs matched, $60$ mW & $0.115$ & $2.74$\\
GaAs matched, $131$ ms & $0.030$ & $3.10$\\
GaAs broadband & $0.049$--$0.052$ & $3.40$--$3.39$\\
QD broadband, $2$ mW & $0.072$ & $3.25$\\
QD broadband, $20$ mW & $0.105$ & $3.11$\\
QD broadband, $30$ mW & $0.117$ & $3.07$\\
QD broadband, $40$ mW & $0.110$ & $3.09$\\
\bottomrule
\end{tabular}
\end{table}

\section{Effective-mode convention and temporal-mode outlook}
\label{sn:modes}
The IID statement refers to effective symbols. All integrated autocorrelation times are in units of the \SI{0.5}{ns} sample: for GaAs $\tau_{\rm int}(X)=1.73$ and $\tau_{\rm int}(X^2)=1.43$; the quantum-dot values $1.63$ and $1.43$ match, the decorrelation being set by the shared receiver low-pass. This fixes the predeclared broadband spacing of two samples, $f_s=\SI{1.0}{GHz}$ and $n_\phi=n/2$, with per-phase statistical floors $0.033$ (GaAs) and $0.036$ (QD).
In the matched mode one bosonic mode is counted per inverse filter bandwidth — the two real quadratures of a carrier mode are not independent source modes — giving $f_s=\SI{10}{MHz}$, half the $2\times10^{7}$ real quadrature degrees of freedom per second that a \SI{10}{MHz} band carries. The energy-test count $n_\phi=n/\tau_{\rm int}(Z^2)=2.54\times10^4$ ($\tau_{\rm int}(Z^2)=103.3$ samples, $2\,621\,440/103.3=25\,377$) follows the measured decorrelation of the squared filtered statistic and is separate from the generation-mode count, the squared statistic decorrelating faster than $Z$ itself. It is a declared effective-sample convention, and the support-disjoint fallback below prices it. {The support-disjoint convention gives $\lfloor2\,621\,440/1001\rfloor=2618$ test symbols on the short record and $\lfloor262\,144\,000/1001\rfloor=261\,882$ on the long \SI{131}{ms} record; the latter gives $\nmax^{\rm test}=0.020$ and $h^*_{\rm cert}=3.17$ bits per symbol for the long record, corresponding to \SI{32}{Mbit/s}, one step below the reported decorrelation-convention value. The matched filter is the single-pass $1001$-tap FIR ($m_h=1001$, hence $q^{\rm match}_{\rm clip}=0.026026$; a forward--backward pass would double the support). The bin width is $\delta=K\,{\rm LSB}_H=0.111$ at $K=64$, {${\rm LSB}_H=(Y_{\max}+Y_{\min})/4096=1.74\times10^{-3}$}, i.e.\ $K$ times the calibrated homodyne LSB of the filtered mode normalised to vacuum variance $\tfrac12$, not the un-normalised filtered voltage. The tabulated $\sigma^2$ in each mode is a one-sided lower confidence bound on the trusted receiver noise: since $\nmax^{\rm test}=\nmax_{\rm true}+(\sigma^2_{\rm true}-\sigma^2_{\rm cal})$ and $c^*$ falls with $\sigma^2$ throughout the calibration interval of each mode, only $\sigma^2_{\rm cal}\le\sigma^2_{\rm true}$ is conservative. The monotonicity is local to that interval and is not claimed globally.}

Two parts of the offline analysis are assumptions rather than measured properties. {First, decorrelation spacing renders the retained outcomes approximately uncorrelated at second order but does not by itself establish independence of the underlying quantum states, a product structure of $\rho_{AE}$ across rounds, or orthonormality of the temporal modes; in particular, adjacent matched-mode outputs at the \SI{10}{MHz} convention share raw samples within the 1001-tap filter support, so the empirical-Bernstein independence hypothesis is a declared modelling convention there, not a proven property; the support-disjoint count of the previous note removes the shared filter support but leaves the product-state assumption in place, and the broadband quantum-dot mode, with two-sample spacing and no long filter, is unaffected. Second, the records are sequential fixed-phase traces, so Protocol~2 is evaluated under stationarity between the phase records, and the offline evaluation draws the energy estimate and the generation statistics from the same phase-$0$ record, whereas a real-time run assigns disjoint test and generation rounds with a private seed before emission. An explicit orthonormal temporal-mode family would remove the first assumption: from the analysis filters $\{g_j\}$ (back-propagated through the calibrated receiver response) one forms the Gram matrix $G_{jk}=\int g_j^*g_k$ and, for $G\succ0$, the L\"owdin waveforms $u_i=\sum_jg_j[G^{-1/2}]_{ji}$, which satisfy $\int u_i^*u_\ell=(G^{-1/2}GG^{-1/2})_{i\ell}=\delta_{i\ell}$ and hence canonical commutators; support-disjoint windows are the simplest special case. Implementing and recalibrating this mode family would replace the matched-mode counting convention by a constructed mode rate. The receiver noise must also be checked for correlations across the new modes because optical-mode orthogonality alone does not guarantee electronic-noise independence.}

\section{Security-theorem proofs and the non-IID setting}
\label{sn:proofs}
\paragraph{Protocol 1.}
The truncation certificate gives the single-round bound $H_{\min}(X|E)\ge h^*_{\rm cert}(\nmax)$ for every admissible $\rho_{AE}$. Conditional min-entropy is exactly additive on product states, so $H_{\min}(X^n|E^n)_{\omega^{\otimes n}}=nH_{\min}(X|E)_\omega\ge n\,h^*_{\rm cert}(\nmax)$, and smoothing can only increase it.

\paragraph{Protocol 2.}
Except with probability $\varepsilon_{\rm test}$ the tested value $\nmax^{\rm test}(T)$ upper-bounds the common per-round mean energy of the IID source. Fix a good accepted transcript. Because the round assignment and phase choices are made independently of the emitted state and before emission, the test measurements and abort condition act only on the test tensor factors, so the conditional state of the generation rounds given $(T{=}t,F{=}1)$ is $\rho_{AE}^{\otimes n_{\rm gen}(t)}$ {(tensored with post-measurement test-side systems that carry no correlation with the generation factors)} with $\Tr(\rho_A\hat n)\le\nmax^{\rm test}(t)$. The single-round bound and additivity then give $H_{\min}(X^{n_{\rm gen}}|E^{n_{\rm gen}},T{=}t,F{=}1)\ge n_{\rm gen}(t)\,h^*_{\rm cert}(\nmax^{\rm test}(t))$; the composable error of the extracted string is abort-inclusive, $\varepsilon_{\rm range}+\varepsilon_{\rm test}+2\varepsilon_s+\varepsilon_{\rm ext}$. {The stated rates are output rates of the seeded protocol: the two-universal hash seed and the round-partition seed are public or external resources in the standard seeded-extraction accounting, and no randomness-expansion claim is made.}

\paragraph{Non-IID setting.}
The IID assumption tensorises the entropy (both protocols) and transfers the test-round energy to the generation rounds (Protocol~2). For product states the asymptotic equipartition property only reproduces the weaker $n\,h^*_{\rm cert}-\sqrt n\,\Delta_{\rm AEP}$, so it brings no advantage without an independent bound on $H(X|E)$ strictly above the min-entropy. A fully non-IID proof would use entropy accumulation, which requires a min-tradeoff function for the energy-constrained discretised homodyne POVM — i.e., a lower bound on $\inf\{H(X|E)_\rho:\Tr(\rho_A\hat n)\le\nmax,\ p_X^\rho=q\}$ for every admissible $q$ — together with a non-IID sampling argument for the two-quadrature test; whether Gaussian states are extremal for this discretised constrained problem is open. A recent energy-bounded SDI generator~\cite{Lu2026} achieves finite-size security for correlated rounds via Kato's inequality in a ternary-input heterodyne protocol; that analysis presupposes within-run randomised preparation and must be re-instantiated for the present phase-switched discretised-homodyne POVM.

\clearpage
\section*{Supplementary Figures}

\begin{figure}[h]
    \centering
    \includegraphics[width=\linewidth]{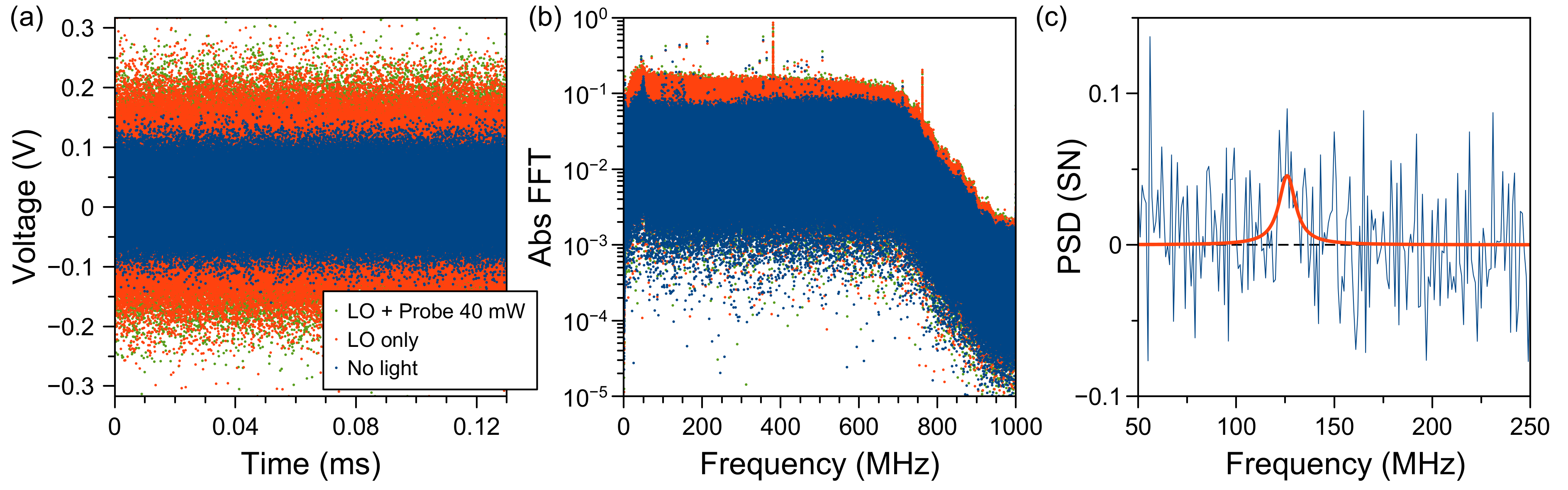}
    \caption{\textbf{Raw records and block-averaged spectra ($n$-GaAs).} (a) ADC time traces at \SI{0.5}{ns} resolution: electronic noise (blue), LO-only (red), LO with \SI{40}{mW} probe scattered light (green, maximising phase); LO power \SI{2}{mW} per diode. (b) Single-trace FFT spectra: the few-percent spin-noise excess is masked by the order-unity fluctuation of the shot-noise floor. (c) PSD in photon-shot-noise units, $(\mathrm{LO{+}Probe}-\mathrm{LO})/(\mathrm{LO}-\mathrm{dark})$, from \SI{1}{\micro\second} blocks averaged over 1310 blocks: the Lorentzian at $f_0=\SI{126.0}{MHz}$ (half-width \SI{5.15}{MHz}) is recovered (red fit), defining the matched mode.}
    \label{fig:adc}
\end{figure}

\begin{figure}[h]
    \centering
    \includegraphics[width=0.85\linewidth]{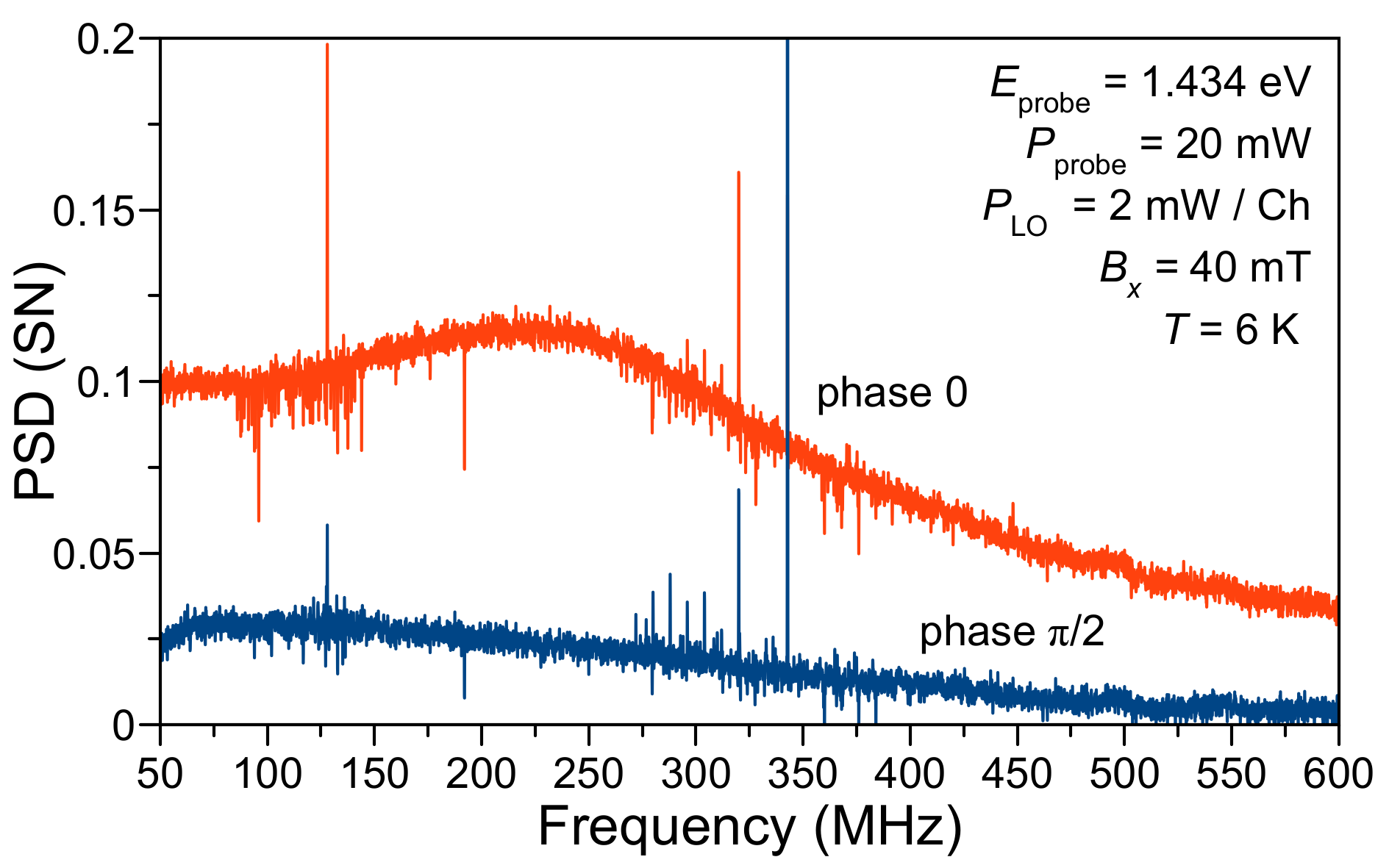}
    \caption{\textbf{Broadband quantum-dot spin noise.} Power spectral density of the electron spin noise in the $(\mathrm{In},\mathrm{Ga})\mathrm{As}$ quantum-dot ensemble at probe power \SI{20}{mW}, in photon-shot-noise units (10-second FFT integration), for the two locked operating phases. The inhomogeneously broadened spectrum fills the detection window up to the analogue cutoff, so every digitised sample carries source energy.}
    \label{fig:qdpsd}
\end{figure}

\begin{figure}[h]
    \centering
    \includegraphics[width=\linewidth]{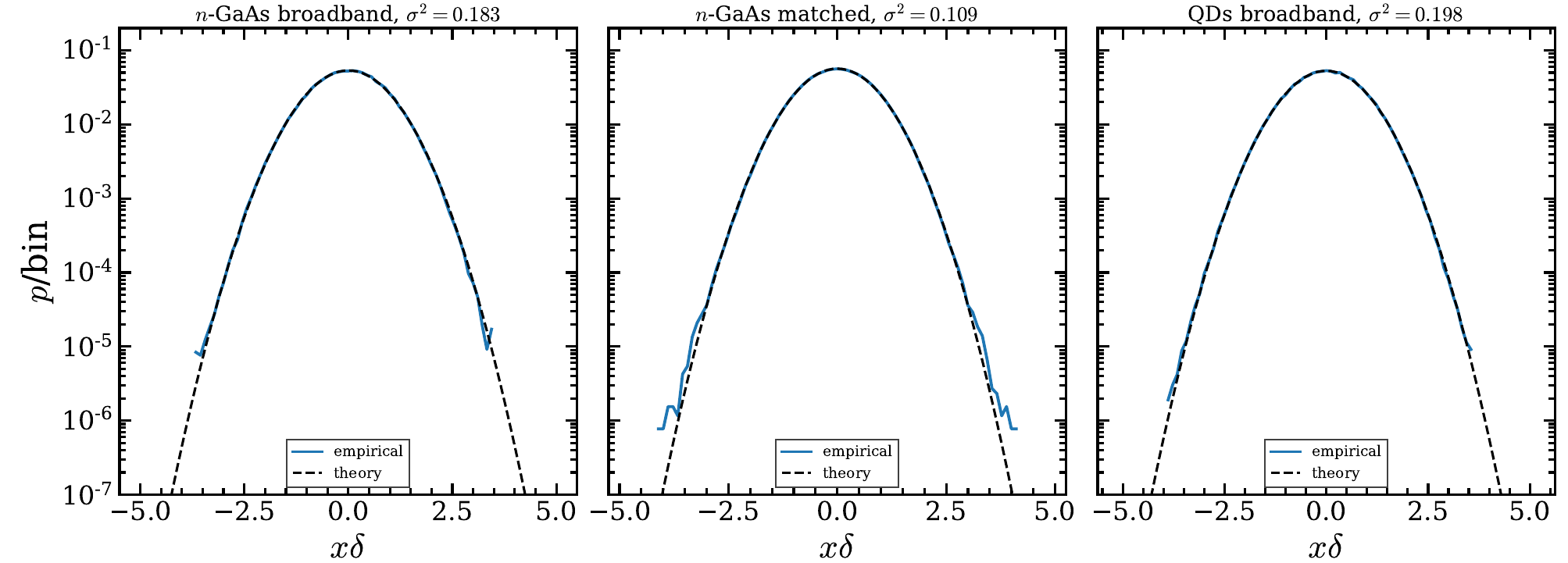}
    \caption{\textbf{Vacuum-reference check of the calibrated POVM.} Empirical binned-symbol distribution of the LO-only reference against the noisy-homodyne POVM diagonal $[\Pi_x]_{00}$ at $\nmax=0$ for the three calibrated receiver modes. The central probability mass agrees with the calibrated vacuum model over the range in which the LO-only statistics are resolved. In the matched mode the empirical tails exceed the Gaussian prediction by a factor of about $1.5$--$2.5$ between $p/\mathrm{bin}\sim10^{-5}$ and $10^{-6}$, and by close to an order of magnitude in the last two resolved bins, which hold only a few counts; these bins carry negligible weight in $\max_x[\Pi_x]_{00}$ and hence in the certificate, but the deviation shows that Gaussianity of the trusted receiver noise is a calibrated model assumption rather than a verified property in the far tails. This is a receiver-calibration test, not a source-randomness claim. {Because the noise variance and receiver zero are calibrated from the same LO-only records, the comparison is an internal consistency check of the receiver model rather than held-out validation.}}
    \label{fig:lohist}
\end{figure}

\begin{figure}[h]
    \centering
    \includegraphics[width=0.85\linewidth]{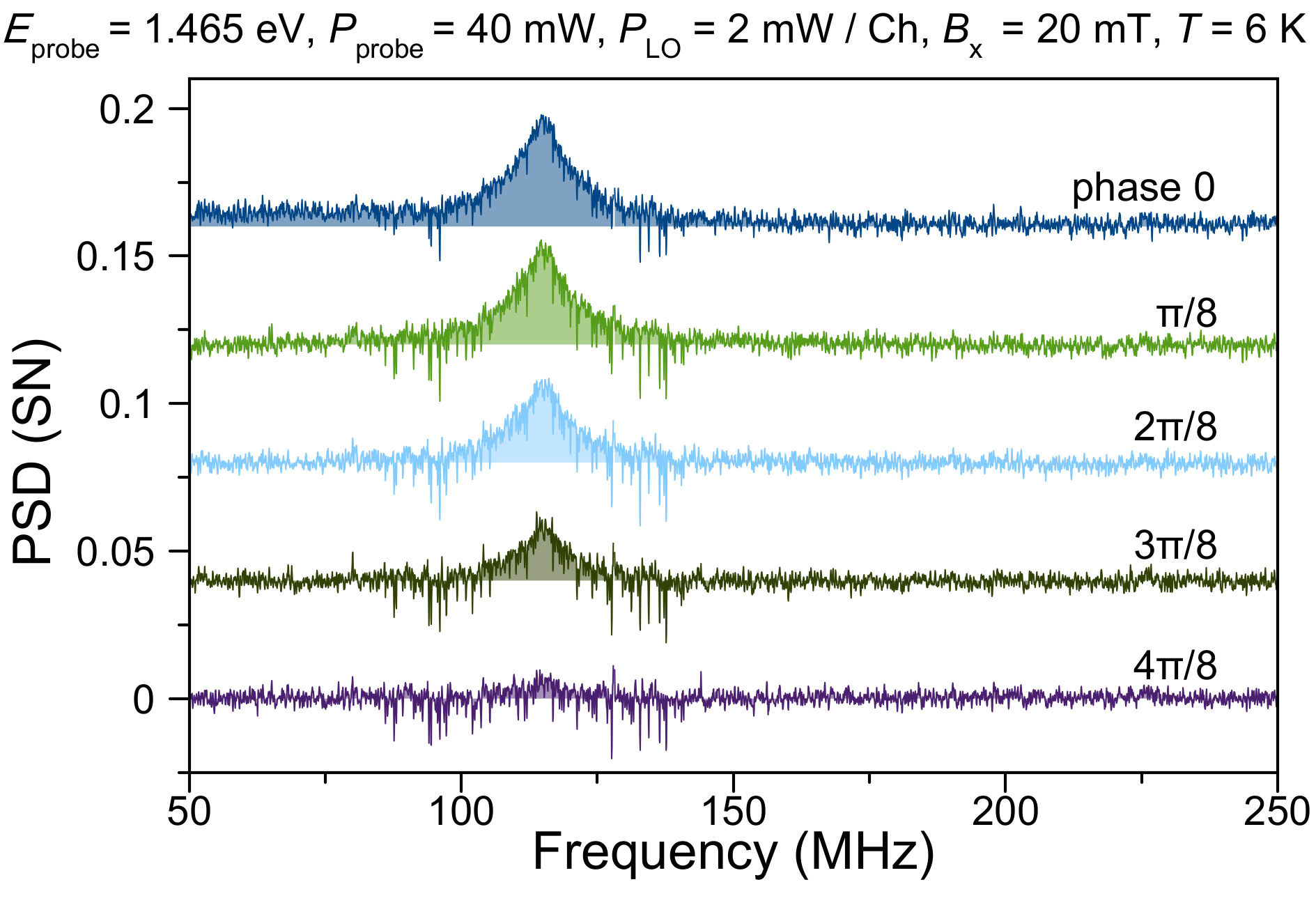}
    \caption{\textbf{Phase dependence of the spin-noise signal.} $n$-GaAs spin-noise PSD at \SI{40}{mW} for five LO--probe phase settings between $\phi=0$ and $\pi/2$ (spectra offset for clarity). The peak follows the expected sinusoidal dependence and is fully suppressed at the second operating point; the scan verifies the phase-sensitive suppression used to lock the test pair, while exact orthogonality ($\bar c=0$) remains a declared calibration (Supplementary Note~6).}
    \label{fig:phase_scan}
\end{figure}

\begin{figure}[h]
    \centering
    \includegraphics[width=0.85\linewidth]{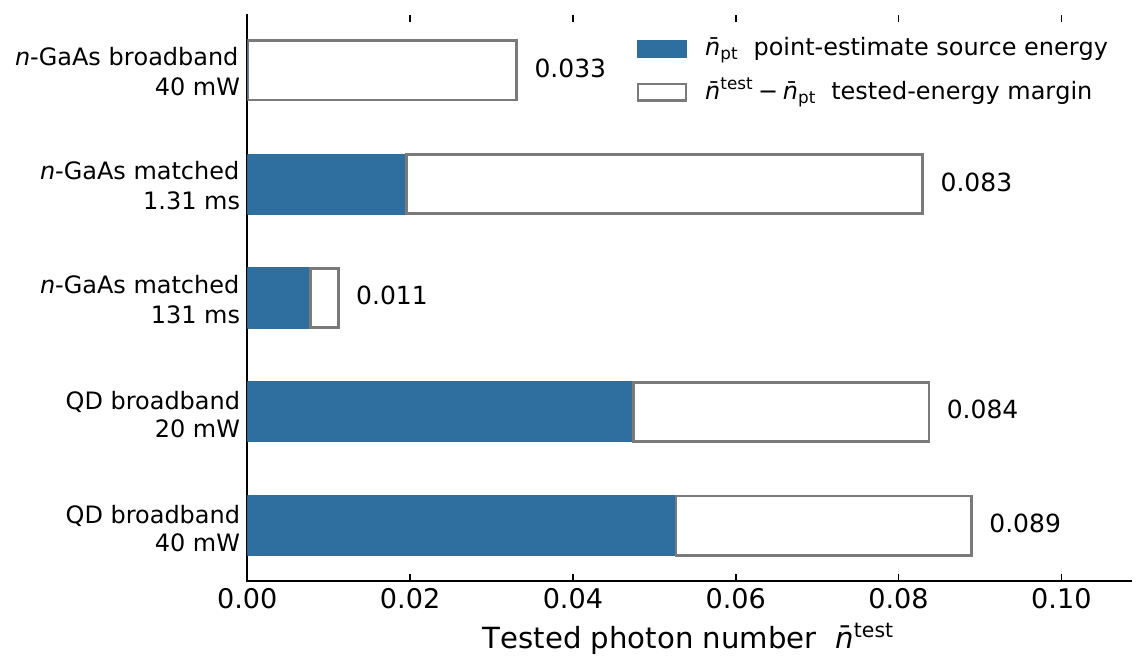}
    \caption{\textbf{Decomposition of the tested energy bound.} Tested photon-number bound $\nmax^{\rm test}$ split into the point-estimate source energy $\nmax_{\rm pt}$ and the statistical/calibration margin $\nmax^{\rm test}-\nmax_{\rm pt}$ for representative operating points. The short-record GaAs matched point and the GaAs broadband benchmark are margin-dominated; the \SI{131}{ms} matched point and the quantum-dot points carry a resolved source-energy contribution.}
    \label{fig:decomposition}
\end{figure}

\begin{figure}[h]
    \centering
    \includegraphics[width=0.85\linewidth]{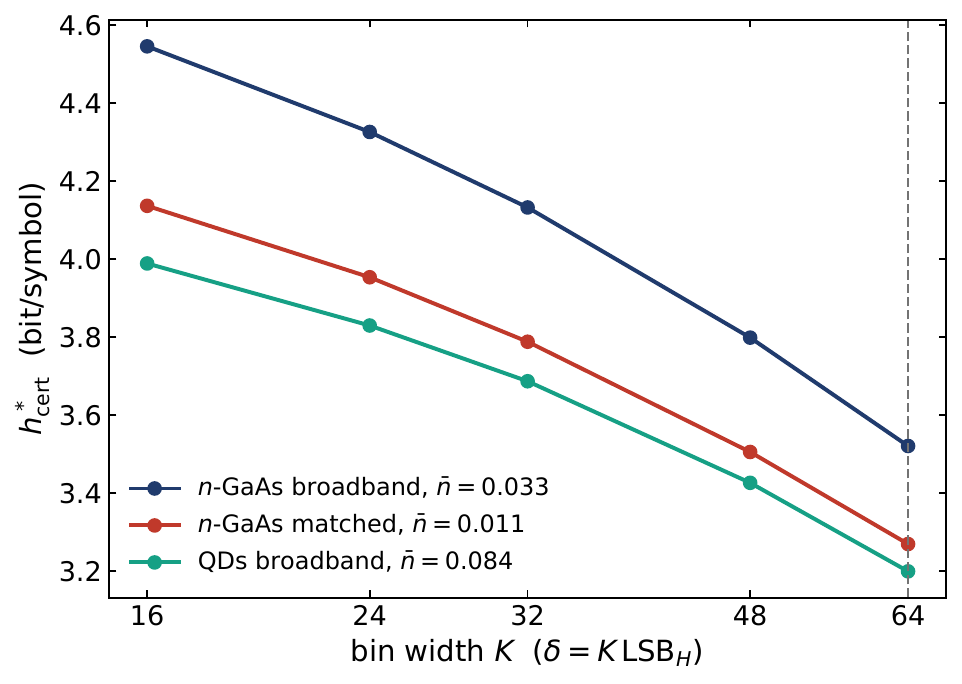}
    \caption{\textbf{Bin-width dependence.} Certified per-symbol entropy versus the bin-width parameter $K$ ($\delta=K\,\mathrm{LSB}_H$) at three operating points; the dashed line marks the fixed operating value $K=64$. Fine binning follows the differential-entropy scaling and coarse binning erodes the certificate, so $K=64$ is an operating choice rather than a tuned post-selection.}
    \label{fig:ksweep}
\end{figure}

\clearpage
\section*{Supplementary Tables}

\begin{table}[h]
\caption{Trusted-receiver calibration. The two GaAs columns use the averaged LO-only convention from the two phase settings; the quantum-dot reference is a single LO-only trace fixing its zero and scale. The matched column applies to the GaAs resonance bandpassed to $126\pm\SI{5}{MHz}$; its energy-test size is the effective squared-filter sample count $n/\tau_{\rm int}(Z^2)$ entering the empirical-Bernstein bound (Supplementary Note~7). {The row $\Delta_{\rm LO}=0$ is exact for the two GaAs columns, where the averaged two-phase LO-only convention cancels the noise-variance asymmetry in the sum; for the single-reference quantum-dot column it holds under the stationarity assumption (one LO reference fixes the QD scale for both phase records), and carrying instead the conservative shared bound $\Delta_{\rm LO}\le\delta_{\rm LO}/(2(1-\bar c))\approx2.5\times10^{-3}$~photons would lower each QD $h^*_{\rm cert}$ by $\lesssim0.014$~bit (largest at the lowest tested energy; re-solved from the same SDP), within the reported precision and leaving the $3.2$--$3.4$~Gbit\,s$^{-1}$ range unchanged.}}
\label{tab:calibration}
\centering
\footnotesize
\begin{tabular}{lccc}
\toprule
 & GaAs broadband & GaAs matched & QD broadband\\
\midrule
$\sigma^2$ & $0.183$ & $0.109$ & $0.198$\\
$Y_{\max}$ & $3.69$ & -- & $3.84$\\
$Y_{\min}$ & $3.42$ & -- & $3.59$\\
$\mathrm{LSB}_H$ (homodyne) & $1.738\times10^{-3}$ & $1.74\times10^{-3}$ & $1.813\times10^{-3}$\\
bin width $\delta=64\,\mathrm{LSB}_H$ & $0.111$ & $0.111$ & $0.116$\\
receiver-zero offset (codes) & $79$ & -- & $69$\\
measured clip fraction $\hat\delta_{\rm clip}$ & $2\times10^{-5}$ & $2\times10^{-5}$ & $2\times10^{-5}$\\
raw clip ceiling $\delta_{\rm clip}^{\max}$ & $2.6\times10^{-5}$ & $2.6\times10^{-5}$ & $2.6\times10^{-5}$\\
filtered clip ceiling $\delta_{\rm clip,filt}^{\max}$ & -- & $4.0\times10^{-8}$ & --\\
generation symbol rate $f_s$ & $1.0$ GHz & $10$ MHz & $1.0$ GHz\\
energy-test size $n_\phi$ & $1.31\times10^{6}$ & $2.54\times10^{4}$ & $1.31\times10^{6}$\\
pre-clip bound $L=2Y_{\max}$ & $7.4$ & -- & $7.7$\\
energy-test range ($Y^2$; $Z^2$ matched) & $Y_{\max}^2=13.6$ & $C=(6\sigma_{\rm filt})^2=22$ & $Y_{\max}^2=14.7$\\
Fock cutoff $\Ncut$ & $200$ & $200$ & $200$\\
$\varepsilon_{\rm test}$ & $10^{-6}$ & $10^{-6}$ & $10^{-6}$\\
$\delta_{\rm LO}$ ($Y^2$) & $4.9\times10^{-3}$ & $4.9\times10^{-3}$ & shared GaAs\\
{per-phase LO-only offset $|d_\phi|$ (codes)} & {$0.31$} & {$0.31$} & {shared GaAs}\\
$\Delta_{\rm LO}$ ($\nmax$) & $0$ & $0$ & $0$\\
\bottomrule
\end{tabular}
\end{table}

\begin{table}[p]
\caption{\textbf{Positioning of the present work.}
Each final-column entry is reported as entropy/rate without rescaling.
The corresponding trust model and certificate are identified using the
abbreviations defined below. ``Theory'' denotes that no experimental
throughput was reported.}
\label{tab:positioning}
\centering
\footnotesize
\setlength{\tabcolsep}{3pt}
\renewcommand{\arraystretch}{1.10}

\begin{tabular*}{\linewidth}{@{\extracolsep{\fill}}lll@{}}
\toprule
Platform and reference
& Model / certificate
& Entropy / rate \\
\midrule

\multicolumn{3}{@{}l}{\textit{Experimental demonstrations}}\\
\addlinespace[2pt]

Alkali SNS \cite{Katsoprinakis2008}
& DD; THR
& n.q. / $\sim 1$ kbit/s \\

Alkali SNS \cite{Dellis2024}
& DD; NIST+TE
& $0.995\pm0.001$ b/out. / $1.04$ Mbit/s \\

Vacuum HD \cite{Gabriel2010}
& DD; QE
& $3.25$ b/samp. / $6.5$ Mbit/s \\

Vacuum HD \cite{Haw2015}
& DD; CMI-C
& $13.75$ b/samp. / $3.55$ Gbit/s \\

Vacuum HetD \cite{Avesani2018}
& SI; HRB
& $13.949$ b/samp. / $17.42$ Gbit/s \\

Sq./thermal HD \cite{Michel2019}
& SI; EUR-Q
& op.-dep. / $8.2$ kbit/s (sq.); $5.2$--$7.2$ kbit/s (th.) \\

Untrusted-light PD \cite{Drahi2020}
& SI; PNT-Q
& $5.32$ b/samp. / $8.05$ Gbit/s \\

Optical SI-QRNG \cite{Cao2016}
& SI; FK
& op.-dep. / $>5$ kbit/s \\

BPSK HD \cite{Rusca2020}
& SDI-E; SMH
& op.-dep. / $145.5$ Mbit/s \\

BPSK HetD \cite{Avesani2020}
& SDI-E; SDP
& $0.09$ b/meas. / $113$ Mbit/s \\

Vacuum HetD \cite{Cizauskas2026}
& SI; HRB
& $12.681$ b/rnd. / $33.92$ Gbit/s \\

This work: GaAs
& SDI-E (RT); F-SDP
& $2.64$--$3.25$ b/mode / $26$--$32$ Mbit/s \\

This work: QD
& SDI-E (RT); F-SDP+SR
& $3.10$--$3.30$ b/mode / $3.10$--$3.30$ Gbit/s \\

\addlinespace[3pt]
\midrule
\multicolumn{3}{@{}l}{\textit{Theory and certification frameworks}}\\
\addlinespace[2pt]

Binary-output CV \cite{VanHimbeeck2017}
& SDI-E; binary
& op.-dep. / theory \\

Multi-outcome CV \cite{Tebyanian2021}
& SDI-E; $d$-outcome
& $\leq 1.585$ b/rnd. / theory \\

Discretised CV \cite{Tebyanian2024SDP}
& SDI-D; F-SDP
& $0.382$ b/rnd. (max.) / theory \\

Squeezed BPSK
\cite{tebyanian2026squeezedstatesemideviceindependentquantumrandomness}
& GO-T; UD; Shannon-C
& $0.199$--$0.516$ b/sym. / theory \\

\bottomrule
\end{tabular*}

\vspace{4pt}
\begin{minipage}{\linewidth}
\scriptsize
\textit{Abbreviations.}
SNS, spin noise; HD, homodyne detection; HetD, heterodyne detection;
PD, photodetection; CV, continuous variable; DD, device-dependent;
SI, source-independent; SDI-E, energy-bounded semi-device-independent;
SDI-D, dimension-bounded semi-device-independent; RT, receiver trusted;
THR, threshold extraction; TE, Toeplitz extraction; QE, quantum-entropy
estimate; CMI-C, conditional min-entropy against classical side
information; HRB, Husimi-resolution bound; EUR-Q, entropic-uncertainty
relation against quantum side information; PNT-Q, photon-number test
against quantum side information; FK, finite-key certificate; SMH,
smooth-min-entropy certificate; SDP, semidefinite-program certificate;
F-SDP, truncated-Fock SDP; SR, signal-resolved; GO-T, trusted
Gram-overlap model; UD, untrusted detector; Shannon-C, Shannon-entropy
certificate against classical side information; n.q., not quantified;
op.-dep., operating-point dependent; sq., squeezed; th., thermal;
b/out., bits per output bit; b/samp., bits per sample; b/meas., bits per
measurement; b/rnd., bits per round; and b/sym., bits per symbol.
\end{minipage}
\end{table}

\clearpage
\bibliographystyle{unsrt}
\bibliography{ref}

\end{document}